\documentclass[online]{tlp}
\usepackage[arxiv]{preamble}
\usepackage{fixkeywords}

\ifarxiv
\usepackage{hyperref}
\fi

\usepackage{amsfonts}
\usepackage{amssymb}
\usepackage{amsmath}
\usepackage{latexsym}
\usepackage{stmaryrd}
\usepackage{microtype}
\usepackage{booktabs}
\usepackage{bigdelim}
\usepackage{comment}
\usepackage{horn} % formatting of listings
\usepackage{customcommands}  % custom commands

\begin{document}
% \lefttitle{Charalambidis et al.}
\lefttitle{A. Charalambidis, B. Kostopoulos, and P. Rondogiannis}

\jnlPage{1}{8}
\jnlDoiYr{2021}
\doival{10.1017/xxxxx}

\title[From Time to Space: The Impact of Linearity in Higher-Order Datalog]{From Time to Space: The Impact of Linearity in Higher-Order Datalog}

\begin{authgrp}
\author{%
      \gn{Angelos} \sn{Charalambidis},
      \gn{Babis} \sn{Kostopoulos}}
\affiliation{Harokopio University of Athens, Greece}
\emails{acharal@hua.gr}{kostbabis@hua.gr}
\author{%
      \gn{Panos} \sn{Rondogiannis}}
\affiliation{National and Kapodistrian University of Athens, Greece}
\email{prondo@di.uoa.gr}
\end{authgrp}
%\history{\sub{xx xx xxxx;} \rev{xx xx xxxx;} \acc{xx xx xxxx}}

%
\maketitle

\begin{abstract}
We consider a fragment of Higher-Order Datalog with negation and argue that it
generalizes the familiar and important fragment of \emph{Linear Datalog}. We
investigate the expressive power of this fragment, establishing a tight
connection with the hierarchy of space complexity classes. In particular, we
demonstrate that for all $k \ge 1$, the $(k+1)$-order fragment of \emph{Stratified Linear
Higher-Order Datalog}$^\neg$ captures $\EXPSPACE[(k-1)]$. This result suggests that
restricting programs to linear recursion shifts the expressive power of the
corresponding fragments from time to space, generalizing the classical result
that (Stratified) Linear Datalog captures \textsf{NL}. Unlike the first-order setting
where an ordering assumption is required to capture $\mathsf{NL}$, our results hold
without any such assumption on the input database. The
proof relies on simulating space-bounded Turing machines using Stratified Linear
Higher-Order Datalog$^\neg$ programs and providing a space-efficient evaluation
of the query program. We argue that identifying such computationally well-behaved
fragments is a crucial step towards paving the way for practical implementations of
Higher-Order Datalog.
\ifarxiv
Under consideration for publication in Theory and Practice of Logic Programming (TPLP).
\fi
\end{abstract}

\section{Introduction}
Over more than two decades, research in logic programming has shifted towards studying and efficiently
implementing fragments that are sufficiently expressive for demanding practical applications. Answer Set
Programming (ASP) exemplifies this trend: despite not being Turing-complete, ASP has found diverse applications
and gained broad acceptance as a robust paradigm. Recent investigations have explored extending
ASP ``beyond $\mathsf{NP}$''~\citep{BJT16Stable-unstablesemanticsBeyondNPnormallogic,ART19BeyondNPQuantifyingoverAnswerSets,FLRSS21PlanningIncompleteInformationQuantifiedAnswerSet}.
More recently, \emph{Higher-Order Datalog with negation}~\citep{iclp24} has been proposed as a powerful
candidate for such an extension. In particular, it has been shown~\citep{iclp25} that $k+1$-order Datalog
with negation, $k\geq 1$, captures $\EXPTIME[k]$ under the well-founded semantics, and captures co-$k$-$\mathsf{NEXPTIME}$
and $k$-$\mathsf{NEXPTIME}$ under the stable model semantics for cautious and brave reasoning, respectively.
However, this power comes at a cost: increased expressiveness inevitably places a significant burden on implementation.
Consequently, identifying a fragment of Higher-Order Datalog with negation that can express problems beyond $\mathsf{NP}$
and that is still amenable to efficient implementation, remains a major challenge.

\begin{table}[h]
\centering
\caption{Expressive power results (entries with a ``$*$'' use the ordering assumption).}
%\caption{Expressive Power of Linear Higher-Order Datalog$^\neg$ (the entries with a ``$*$'' hold under the ordering assumption).}
\label{table1}
{\tablefont\begin{tabular}{@{\extracolsep{\fill}}p{5.9cm}cccccl@{}}
\topline
{\bf Fragment} &  \multicolumn{4}{c}{\bf Order of the program} \\
\cmidrule(lr){2-5}
 & 1 & 2 & $\cdots$ & $k+1$
\midline
Higher-Order Datalog$^{\neg}$  & \textsf{P}$^*$ & \textsf{EXPTIME} & $\cdots$ & $k$-\textsf{EXPTIME} \\
Stratified Higher-Order Datalog$^{\neg}$ &  &  &  &  &
\midline
Stratified Linear Higher-Order Datalog$^{\neg}$ & \textsf{NL}$^*$ & \textsf{PSPACE} & $\cdots$ & $(k-1)$-\textsf{EXPSPACE}
\botline
\end{tabular}}
\end{table}

A promising avenue of research is the study of a \emph{linear} fragment of Higher-Order Datalog with negation.
In the first-order Datalog setting, \emph{linearity} is a pivotal property defined by restricting every rule
body to contain \emph{at most} one recursive atom~\cite[Chapter~15]{Ullmann}. This restriction is practically
significant for two reasons. First, it is widely recognized that Linear Datalog suffices to express most ``real-life''
recursive queries, see, e.g.,~\citep[p.~405]{Consens1990}. Second, from an implementation standpoint, Linear Datalog
is far more amenable to specialized optimizations than full Datalog~\citep{Ullmann}. Theoretically, this efficiency
is formalized by the fact that Linear Datalog captures $\mathsf{NL}$ (Nondeterministic Logarithmic Space) on ordered
databases~\citep{dantsin}. Intuitively, this implies that a query can be evaluated using logarithmic space relative
to the input size. In practical terms, the evaluation engine need not store the entire history of derivations,
but only the current state of the computation (e.g., the ``current node'' in a traversal). This
favorable complexity profile extends to programs with negation: it has been shown that \emph{Stratified Linear Datalog}
- where rules are limited to linear recursion within the same stratum -
also captures $\mathsf{NL}$ on ordered databases~\citep{Consens1990}.
It is therefore interesting (and potentially practical) to investigate what linearity
means in a higher-order setting and its precise expressive power.

In this paper, we undertake the study of the above question.
We generalize the concept of Stratified Linear Datalog to the higher-order setting,
getting a language which in the rest of the paper we will refer to as
\emph{Stratified Linear Higher-Order Datalog}$^\neg$.
Our extension is based on the notion of stratification introduced by~\cite{iclp24},
which ensures programs are stratified not only with respect to negation but also with
respect to higher-order application.
Our central contribution proves that for all $k \ge 1$, the
$(k+1)$-order fragment of Stratified Linear Higher-Order Datalog$^\neg$ captures $\EXPSPACE[(k-1)]$.
Unlike the first-order setting where an ordered database is required to capture
$\mathsf{NL}$, our higher-order characterization holds without any such ordering
assumption. Since $\EXPSPACE[k]$ is a subset of the $\EXPTIME[k]$ class
captured by the unrestricted Higher-Order Datalog$^\neg$ language~\citep{iclp25}, our results
indicate that, potentially, linearity leads to a more ``manageable''
language from an implementation point of view. Table~\ref{table1} shows more compactly the
aforementioned results: the second row is new, the first one is from~\citep{iclp25}.
Surprisingly, our results have a direct correspondence with Neil Jones' results~\citep{Jones}
on the expressive power of the \emph{tail-recursive} fragments of a restricted
higher-order functional language. Given the close conceptual link between tail-recursion and
linearity, it is remarkable but not totally unexpected that these distinct programming paradigms
capture identical complexity hierarchies.

In conclusion, the main contributions of the present paper can be summarized as follows:
\begin{itemize}
\item {\bf Extension of Linearity}: We extend the classical notion of \emph{linearity} to the class
      of higher-order Datalog programs. We believe that this extension establishes the basis for developing
      novel optimizations for higher-order logic programming.

\item {\bf Lower Bound Simulation}: We provide a simulation of space-bounded Turing machines
      using Stratified Linear Higher-Order Datalog$^\neg$, using higher-order
      predicates to encode the counting of the tape-space of the Turing machine, its configurations, and
      its execution. The simulation is natural and does not require any tedious encodings.
      This simplicity suggests that the proposed language is not merely a theoretical construct, but a natural
      and intuitive fragment of Higher-Order Datalog$^\neg$.

\item {\bf Upper Bound and Space Efficient Evaluation}: We present a space-efficient
      proof procedure for computing queries. By using top-down evaluation, we avoid the
      full program interpretation, which would otherwise
      consume space that is tower-exponential in the order of the program. We believe that a refinement
      of our proof procedure could lead to a viable implementation of
      the proposed language.
\end{itemize}
Sections~\ref{preliminaries} and~\ref{linear-datalog} introduce the syntax, semantics, and motivation for
Stratified Linear Higher-Order Datalog$^\neg$. Sections~\ref{simulation}, \ref{computation}, and the
\ifincludeappendix Appendix \else supplementary material \fi
establish theoretical space complexity bounds and the correctness proof of the space-efficient
evaluation algorithm.

% While HO-Datalog has been studied for its semantic properties and utility in functional-logic programming,
% a systematic classification of its expressive power with respect to the standard complexity hierarchy has been lacking.
% Specifically, as we ascend the hierarchy of types (from individuals to sets, sets of sets, etc.),
% does the expressive power of the language grow correspondingly? And if so, exactly which complexity classes are captured at each order?

% In this paper, we focus on the \emph{Linear} fragment of HO-Datalog$^\neg$.
% By restricting programs to linear recursion (a single recursive call per rule formulation), we demonstrate
% that the expressive power shifts from Time to Space. Specifically, we prove that $(k+1)$-order
% Linear HO-Datalog$^\neg$ captures \EXPSPACE[(k-1)] , mimicking the relationship
% where First-Order Linear Datalog captures \textsf{NL}~ (conceptually related to space-bounded paths).

% Our proofs rely on simulating higher-order Turing Machines using higher-order predicates to encode configurations and tape contents.
% By defining appropriate "counting modules"---higher-order predicates that represent integers of exponential magnitude---we can index the vast space
% of these complexity classes purely within the logic. These results clarify the theoretical standing of Linear Higher-Order Datalog and provide a
% framework for understanding the trade-offs between higher-order abstraction and computational resources.

\section{Higher-Order Datalog with Negation: Preliminaries}\label{preliminaries}

% \subsection{Syntax}

In this section we introduce the syntax of Higher-Order Datalog$^\neg$, following~\cite{iclp24}.
% For simplicity reasons, the syntax of $\HOL$ does not
% include function symbols; this is a restriction that can easily be lifted.
The language uses two base types: $\bool$, the Boolean domain, and $\basedom$,
the domain of data objects. The composite types are partitioned into
\emph{predicate} ones (assigned to predicate symbols) and \emph{argument} ones
(assigned to parameters of predicates).
\begin{definition}
  Types are either \emph{predicate} or \emph{argument}, denoted by $\pi$
  and $\rho$ respectively, and defined as:
  \begin{align*}
    %\sigma & := \iota \mid (\iota \rightarrow \sigma) \\
    \pi  & := \bool \mid (\rho \to \pi) \\
    \rho & := \basedom \mid \pi
  \end{align*}
\end{definition}

The binary operator $\to$ is right-associative. Every predicate type $\pi$ can be written in the form
$\rho_1 \to \cdots \to \rho_n \rightarrow \bool$,
$n\geq 0$ (for $n=0$ we assume that $\pi=\bool$).
%
% We proceed by defining the syntax of $\HOL$.
%

\begin{definition}
  The \emph{alphabet} of Higher-Order Datalog$^\neg$ consists of: \emph{predicate variables}
  of every predicate type $\pi$ (denoted by capital letters such as $\mathsf{P},\mathsf{Q},\ldots)$;
  \emph{predicate constants} of every predicate type $\pi$ (denoted by lowercase letters such as $\mathsf{p},\mathsf{q},\ldots$);
  \emph{individual variables} of type $\basedom$ (denoted by capital letters such as $\mathsf{X},\mathsf{Y},\ldots$);
  \emph{individual constants} of type $\basedom$ (denoted by lowercase letters such as $\mathsf{a},\mathsf{b},\ldots$);
  the \emph{equality}  constant $\approx$ of type $\basedom \to \basedom \to o$;
  the \emph{conjunction} constant $\wedge$ of type $\bool \to \bool \to \bool$;
  the \emph{inverse implication} constant $\lrule$ of type $\bool \to \bool \to \bool$; and
  the \emph{negation} constant ${\tt not}$ of type $\bool \to \bool$.
\end{definition}

Arbitrary variables (either predicate or individual ones) will be denoted by $\mathsf{R}$. % and its subscripted versions.% Save a line

\begin{definition}
  \label{def:expressions}

  The \emph{expressions} and \emph{literals} of Higher-Order Datalog$^\neg$ are
  defined as follows.
  Every predicate variable/constant and every individual variable/constant is an
  expression of the corresponding type; if $\mathsf{E}_1$ is an expression of type
  $\rho \to \pi$ and $\mathsf{E}_2$ an expression of type $\rho$ then
  $(\mathsf{E}_1\ \mathsf{E}_2)$ is an expression of type $\pi$.
  Every expression of type $\bool$ is called an \emph{atom}.
  If $\mathsf{E}$ is an atom, then $\mathsf{E}$ and $({\tt not}\, \mathsf{E})$
  are literals of type $\bool$; if $\mathsf{E}_1$ and
  $\mathsf{E}_2$ are expressions of type $\basedom$, then $(\mathsf{E}_1\approx
    \mathsf{E}_2)$ and ${\tt not}\, (\mathsf{E}_1\approx \mathsf{E}_2)$
  are literals of type $\bool$.
\end{definition}

We will omit parentheses when no confusion arises.
% To denote that an expression
% $\mathsf{E}$ has type $\rho$ we will often write $\mathsf{E}:\rho$.

\begin{definition}
  A \emph{rule} of Higher-Order Datalog$^\neg$ is a formula
  $\mathsf{p}\ \mathsf{R}_1 \cdots \mathsf{R}_n \lrule \mathsf{L}_1 \land \ldots \land \mathsf{L}_m$,
  where $\mathsf{p}$ is a predicate constant of type $\rho_1 \to \cdots \to \rho_n \to \bool$,
  $\mathsf{R}_1,\ldots,\mathsf{R}_n$ are distinct variables of types $\rho_1,\ldots,\rho_n$ respectively and
  the $\mathsf{L}_i$ are literals.
  The literal $\mathsf{p}\ \mathsf{R}_1 \cdots \mathsf{R}_n$ is the \emph{head} of the rule and
  $ \mathsf{L}_1 \land \ldots \land \mathsf{L}_m$ is the \emph{body} of the rule.
  A \emph{program} $\Prog$ of Higher-Order Datalog$^\neg$ is a finite set of rules.
\end{definition}

We will write
$\mathsf{L}_1,\ldots,\mathsf{L}_m$ instead of
$\mathsf{L}_1 \wedge \cdots \wedge \mathsf{L}_m$ for the body of a rule.
For brevity, we will often denote
a rule as $\mathsf{p} \ \overline{\mathsf{R}} \lrule \mathsf{B}$, where
$\overline{\mathsf{R}}$ is a shorthand for a sequence of variables
$\mathsf{R}_1 \cdots \mathsf{R}_n$ and $\mathsf{B}$ represents the body of the rule.
We will avoid using currying as much as possible and
will use tuples instead, a syntax that is more familiar to logic programmers.
The tuple syntax can be directly transformed to the curried one by a simple
preprocessing. So, for example, instead of {\tt succ Ord X Y} we will write {\tt succ(Ord,X,Y)},
instead of the partial application {\tt succ Ord} we will write {\tt succ(Ord)},
and so on. More generally, the partial application
$\mathsf{p}\ \mathsf{E}_1\ \cdots\ \mathsf{E}_n$ will be written as
$\mathsf{p}(\mathsf{E}_1,\ldots,\mathsf{E}_n)$.
\begin{example}\label{subset-example}
The following is a Higher-Order Datalog$^\neg$ program defining the \texttt{subset} predicate of
type $(\iota \to o) \to (\iota \to o) \to o$: it takes two first-order unary relations (ie., two sets)
as arguments and checks if the first is a subset of the second:
\begin{lstlisting}
subset(P,Q):-not nonsubset(P,Q).
nonsubset(P,Q):-P(X),not Q(X).
\end{lstlisting}
Note how we implement universal quantification in the body
of a rule: to express that \texttt{P} is a subset of \texttt{Q}, i.e.,
that every \texttt{X} that belongs to \texttt{P} also belongs to \texttt{Q}, we just require
that it is not the case that \texttt{P} is a \texttt{nonsubset} of \texttt{Q}, i.e., it
is not the case that there exists \texttt{X} that belongs to the first relation but not the other.\hfill\hbox{\proofbox}
\end{example}
The formal semantics of Higher-Order Datalog$^\neg$~\cite{iclp24} is detailed in the
\ifincludeappendix appendix \else supplementary material \fi, but is not
required to follow this paper. The programs considered here can be understood
purely declaratively. The fragments of Higher-Order Datalog$^\neg$ we investigate are
semantically well-behaved; the predicates we employ simply denote extensional higher-order relations. For
instance, a unary second-order predicate denotes a relation that accepts a classical set as an argument
and evaluates to $\mtrue$ or $\mfalse$.

The notion of \emph{order} of a predicate, is formally defined as follows:
\begin{definition}
  The \emph{order} of a type is recursively defined as follows:
  \[
    \begin{array}{rcl}
      \textit{order}(\iota)        & = & 0                                \\
      \textit{order}(o)            & = & 1                                \\
      \textit{order}(\rho \to \pi) & = & \max\{ \mathit{order}(\rho) + 1,
      \mathit{order}(\mathit{\pi}) \}                                     \\
      % \textit{order}(\rho_1 \rightarrow \cdots \rightarrow \rho_n \rightarrow o) & = & 1+\textit{max}(\{\textit{order}(\rho_i) \mid 1 \leq i \leq n\})
    \end{array}
  \]
  The order of a predicate constant (or variable) is the order of its type.
\end{definition}
\begin{definition}
  For all $k\geq 1$, \emph{$k$-Order Datalog$^\neg$} is the fragment of Higher-Order Datalog$^\neg$
  in which all variables have order less than or equal to $k-1$ and all predicate constants
  in the program have order less than or equal to $k$.
\end{definition}

\section{Stratified Linear Higher-Order Datalog$^\neg$}\label{linear-datalog}
In this section we formally define the class of \emph{Stratified Linear Higher-Order Datalog}$^\neg$
programs and illustrate it by examples. This class is a proper extension of the class of
\emph{Stratified Linear Datalog} programs, initially studied by~\cite{Consens1990}.

We need a notion of stratification for higher-order programs. In analogy to Datalog, there
exists a natural such notion for Higher-Order Datalog$^\neg$~\citep{iclp24}.
\begin{definition}\label{stratified}
A program $\Prog$ is called \emph{stratified} if there exists a function $S$ mapping
predicate constants to natural numbers, such that for each rule
$\mathsf{p} \ \overline{\mathsf{R}} \leftarrow \mathsf{L}_1,\ldots,\mathsf{L}_m$
and any $i\in\{1,\ldots, m\}$:
\begin{itemize}
  \item $S(\mathsf{q})\leq S(\mathsf{p})$ for every predicate constant
        $\mathsf{q}$ occurring in $\mathsf{L}_i$.
  \item If $\mathsf{L}_i$ is of the form $({\tt not}\,\,\mathsf{E})$, then
        $S(\mathsf{q})< S(\mathsf{p})$ for each predicate constant $\mathsf{q}$
        occurring~in~$\mathsf{E}$.
  \item For any subexpression of $\mathsf{L}_i$ of the form
        $(\mathsf{E}_1~\mathsf{E}_2)$, $S(\mathsf{q})< S(\mathsf{p})$ for every
        predicate constant $\mathsf{q}$ occurring in $\mathsf{E}_2$.
\end{itemize}
\end{definition}

A distinctive feature of this definition is the third condition; the intuition behind this
constraint is that Higher-Order Datalog$^\neg$ allows for the definition of higher-order predicates that simulate
negation (e.g., via a rule such as \lstinline|neg P :- not P|). Consequently, to ensure
a well-defined semantics, predicate constants occurring as arguments in a higher-order
application must be treated as if they occur within a negative context.

As shown in~\cite{iclp24}, stratified Higher-Order Datalog$^\neg$ programs have a unique stable
model which coincides with their well-founded model. For this reason, such programs are particularly
appealing from a semantic point of view.

We can now define the exact class of programs that we will be studying.
\begin{definition}[Stratified Linear Higher-Order Datalog$^\neg$]
A Higher-Order Datalog$^\neg$ program  $\mathsf{P}$ is called \emph{stratified linear} if it is
stratified with respect to a stratification function $S$ and, for every rule
$\mathsf{p}\ \mathsf{R}_1 \cdots \mathsf{R}_n \leftarrow \mathsf{L}_1, \ldots, \mathsf{L}_m$
in $\mathsf{P}$, the following conditions hold:
\begin{enumerate}
\item There exists at most one predicate constant $\mathsf{q}$ occurring in the body
     $\mathsf{L}_1, \ldots, \mathsf{L}_m$ such that $S(\mathsf{q}) = S(\mathsf{p})$.

\item If such a predicate $\mathsf{q}$ exists, it must appear in exactly one literal $\mathsf{L}_i$.
\end{enumerate}
The class of all Stratified Linear Higher-Order Datalog$^\neg$ programs is called
\emph{Stratified Linear Higher-Order Datalog}$^\neg$.
\end{definition}

This condition implies that a rule can access any number of atoms from lower strata - treating them
as already computed data - but can pass the active computation state to at most one literal in the current stratum.

%The following example, adopted from~\citep{iclp25}, illustrates the core features of the language and defines
%auxiliary predicates essential for our subsequent simulations. Notably, this example was originally developed in a
%different context, yet it fully complies with the restrictions of Stratified Linear Higher-Order Datalog$^\neg$.
%This underscores the naturalness of the proposed fragment.
%
We now give two examples of well-known problems that can be written very concisely in
Stratified Linear Higher-Order Datalog$^\neg$; this underscores the naturalness and usefulness
of the fragment.
\begin{example}\label{example1}
We define the relation
\lstinline`hamilton(X,Y)` which is true iff there exists a Hamilton path from vertex
\lstinline`X` to vertex \lstinline`Y` in a graph represented by a binary predicate
\lstinline`e` which specifies the edges of the graph.
The first rule in the definition of \lstinline`hamilton` is the following:
\begin{lstlisting}
hamilton(X,Y):-ordering(Ord),first(Ord,X),last(Ord,Y),subset(succ(Ord),e).
\end{lstlisting}
The above rule states that there exists a Hamilton path from vertex
\lstinline`X` to vertex \lstinline`Y` if there exists a relation \lstinline`Ord`
that is a strict total ordering with first element \lstinline`X` and last element
\lstinline`Y` and for every two consecutive elements in \lstinline`Ord` the
corresponding edge exists in \lstinline`e`. We use the predicate \lstinline`subset` of Example~\ref{subset-example} modified to handle binary relations. Notice the use of \lstinline`Ord` in
the above rule: it is a predicate variable that does not appear in the head of
the rule and therefore it is an existentially quantified variable of the body
(\ie the body can be read as ``there exists a relation \lstinline`Ord` such that
\ldots''). To be a strict total ordering, \lstinline`Ord` must be irreflexive,
transitive, and every two different elements must be related. This can be
expressed with the following rules:
\begin{lstlisting}
ordering(Ord):-connected(Ord),transitive(Ord),irreflexive(Ord).
connected(Ord):-not disconnected(Ord).
disconnected(Ord):-not Ord(X,Y),not Ord(Y,X),not(X=Y).
transitive(Ord):-not non_transitive(Ord).
non_transitive(Ord):-Ord(X,Y),Ord(Y,Z),not Ord(X,Z).
irreflexive(Ord):-not non_irreflexive(Ord).
non_irreflexive(Ord):-Ord(X,X).
\end{lstlisting}
%
%It is interesting to note above how we can implement
%universal quantification in the body of a rule: for example, to express the fact
%that \lstinline`Ord` is connected, \ie that \emph{for all} \lstinline`X`,
%\lstinline`Y` either \lstinline`X` is related to \lstinline`Y` or vice-versa, we
%just require that \lstinline`Ord` is not disconnected, \ie it is not the case
%that there exist \lstinline`X`, \lstinline`Y` that are not related. This is a
%common trick that we will use throughout the paper in order to represent
%universal quantification.

We now define the predicates \lstinline`first`, \lstinline`last` and \lstinline`succ`.
Predicate \lstinline`first(Ord,X)` is true for \lstinline`X` being the individual constant that
is the first element with respect to the ordering specified by \lstinline`Ord`.
Likewise, \lstinline`last(Ord,X)` is true if \lstinline`X` is the last element
in \lstinline`Ord`. The predicate \lstinline`succ(Ord,X,Y)` is true for \lstinline`X`
and \lstinline`Y` that are sequential in \lstinline`Ord`.
\begin{lstlisting}
first(Ord,X):-not nfirst(Ord,X).
nfirst(Ord,X):-Ord(Z,X).
last(Ord,X):-not nlast(Ord,X).
nlast(Ord,X):-Ord(X,Y).
succ(Ord,X,Y):-Ord(X,Y),not nsequential(Ord,X,Y).
nsequential(Ord,X,Y):-Ord(X,Z),Ord(Z,Y).
\end{lstlisting}
It is easy to check that this program is a Stratified Linear Higher-Order Datalog$^\neg$ one.
\hfill\hbox{\proofbox}
\end{example}
\begin{example}\label{graph_recoloring_example}
The \emph{Graph Recoloring} problem asks whether it is possible
to transform one valid vertex coloring of a graph into another valid target coloring by changing the color of
exactly one vertex at a time. A crucial constraint is that every intermediate
step in this sequence must also be a valid coloring (\ie no two adjacent
vertices may share the same color). Because the configuration space of all possible colorings for a graph is finite
but exponentially large, searching for a path between two colorings is an
archetypal reachability problem. \cite{Bonsma2009} prove that
deciding whether such a reconfiguration path exists between two colorings is a
\textsf{PSPACE}-complete problem.

The following program solves the Graph Recoloring problem. In this program, a
graph is defined by the predicates \lstinline`node` and \lstinline`edge`.
A specific coloring of a graph is represented as a binary relation \lstinline`C`,
where the atom \lstinline`C(X,Col)` indicates that node \lstinline`X`
is assigned color \lstinline`Col`. The target coloring is represented by
the relation \lstinline`target_coloring`.

Before we can determine if a graph can be recolored, we must first define what
constitutes a structurally sound, valid coloring. We do this by defining the
conditions that make a coloring invalid, and then concluding that a coloring is
valid if none of those conditions apply.
\begin{lstlisting}
has_color(C,X) :- C(X,Col).
invalid_coloring(C) :- edge(X,Y),C(X,Col),C(Y,Col).
invalid_coloring(C) :- C(X,Col1),C(X,Col2),not (Col1=Col2).
invalid_coloring(C) :- node(X), not has_color(C,X).
valid_coloring(C) :- not invalid_coloring(C).
\end{lstlisting}
The \lstinline`invalid_coloring` predicate captures the three ways a coloring attempt can fail:
(a) two distinct nodes connected by an edge share the exact same color;
(b) a single node is simultaneously assigned two different colors, and
(c) a node in the graph does not have any color assigned to it at all.

A valid recoloring sequence requires changing the color of exactly one node at a
time, ensuring that every intermediate state remains a valid graph coloring.
Finally, the program must check if this sequence of transitions successfully
reaches the \lstinline`target_coloring`.
\begin{lstlisting}
valid_step(C1,C2) :- valid_coloring(C1),valid_coloring(C2),one_diff(C1,C2).
recolorable(C) :- eq(C,target_coloring).
recolorable(C) :- valid_step(C,NextC),recolorable(NextC).
\end{lstlisting}
The \lstinline`valid_step(C1,C2)` enforces the rules
that both the starting coloring (\lstinline`C1`) and
resulting coloring (\lstinline`C2`) are valid colorings,
and exactly one node changes color. The latter is checked by the
predicate \lstinline`one_diff`. Both \lstinline`one_diff` and
the auxiliary predicate \lstinline`eq` that checks whether two binary
relations are equal can be defined as follows.
\begin{lstlisting}
diff(C1,C2,X) :- C1(X,Col), not C2(X,Col).
has_diff(C1,C2) :- diff(C1,C2,X).
multiple_diffs(C1,C2) :- diff(C1,C2,X),diff(C1,C2,Y), not (X=Y).
one_diff(C1,C2) :- has_diff(C1,C2),not multiple_diffs(C1,C2).
eq(C1,C2) :- not has_diff(C1,C2),not has_diff(C2,C1).
\end{lstlisting}
It is easy to verify that the above program is stratified:
every non-recursive predicate in the program can be assigned on its own stratum,
and the sole recursive predicate, \lstinline`recolorable`,
depends only positively on itself, therefore preserving this hierarchy.
Furthermore, the program maintains the linearity within every stratum.
The only interesting case is the recursive rule of predicate \lstinline`recolorable`,
where the body contains exactly one predicate from the same stratum
that is \lstinline`recolorable` itself. This predicate appears exactly once,
alongside \lstinline`valid_step` which resides in a lower stratum.\hfill\hbox{\proofbox}
\end{example}

Languages such as Higher-Order Datalog$^\neg$, are usually referred as
\emph{formal query languages}. A program in our language can be considered to
compute a query in the following sense: a first-order predicate, like
\lstinline`e` in Example~\ref{example1}, will be called an
\emph{input predicate} and its denotation (as a set of ground atoms) constitutes what is
called the \emph{input database}, usually denoted by $D_{in}$; a first-order
predicate like \lstinline`hamilton` in Example~\ref{example1}, will be an
\emph{output} one and its denotation constitutes the \emph{output database},
usually denoted by $D_{out}$.

Formally, a \emph{database schema} $\sigma$ is a finite set of first-order
predicate symbols with associated arities. A \emph{database} over a schema $\sigma$
is a finite set of ground atoms whose predicate symbols belong to $\sigma$.
A \emph{query} is a mapping from databases over a schema $\sigma_1$ to databases over a schema $\sigma_2$.
% It suffices to assume that $\sigma_2$ is singleton, \ie there is a single output predicate.
%
A program $\mathsf{P}$ can be seen as a query ${\cal Q}_{\mathsf{P}}$ such that
$D_{out} = {\cal Q}_{\mathsf{P}}(D_{in})$. We are interested in queries that are
\emph{generic}~\citep{Imm86}, \ie queries that do not depend on the names of the
individual constants in the input database.
% Notice that all these notions can be easily generalized to databases of an arbitrary
% but fixed set of predicate symbols of arbitrary but fixed arities~\citep{DBLP:journals/jcss/Schlipf95}.
Given a fragment of our language, we are interested in the \emph{expressive power}
of the fragment, namely the set of queries that can be defined by
programs of the fragment. We want to demonstrate that the
fragment \emph{captures a complexity class $\mathcal{C}$}, \ie it can express
\emph{exactly} all the queries whose \emph{evaluation complexity} belongs to
$\mathcal{C}$. Evaluation complexity is the complexity of checking
whether a given atom belongs to the output database.
%Notice that different
%semantics of the fragments we study may lead to different evaluation
%complexities and therefore to capturing different complexity classes. Therefore,
%our results will be of the form ``\emph{the fragment X of Higher-Order
%Datalog$^\neg$, under the Y semantics, captures the complexity class Z}''.

\section{Simulation of Space-bounded Turing Machines\label{simulation}}
In this section, we demonstrate how any query that belongs to  \EXPSPACE[(k-1)] ($k\geq 1$), can be
expressed by a $(k+1)$-order Stratified Linear Datalog$^\neg$ program.
Without loss of generality, we assume that the output schema of the query consists of
a single output predicate, since every query can be decomposed into multiple
queries of this form. By definition, since the query belongs to \EXPSPACE[(k-1)], there exists a
Turing machine that given on its tape an input database under some sensible encoding,
decides whether a tuple belongs to the output relation of
the query; furthermore, it does so using at most $\exp_{k-1}(n^d)$ tape cells, where $n$ is the number of constant
symbols in the input database and $d$ is some sufficiently large constant
(assume $\exp_{-1}(x) = \lceil \log x \rceil$).
We formalize this result as follows:
\begin{theorem}\label{lower-bound-theorem}
Every query in \EXPSPACE[(k-1)] ($k \geq 1$) can be expressed by a $(k+1)$-order
Stratified Linear Datalog$^\neg$ program.
\end{theorem}

To establish this theorem, we construct a simulation of an arbitrary
space-bounded Turing machine using Stratified Linear Higher-Order
Datalog$^\neg$. Our simulation strategy is to represent the tape as a function
that maps an address of a tape cell to a symbol of the alphabet of the machine. To uniquely
identify cells on an exponentially large tape, we require a mechanism to
represent and manipulate equally large natural numbers as tape indices. We
utilize higher-order variables to represent numbers; specifically, we construct
``counting modules'' where a predicate of order $k$ can represent a
number up to $\exp_{k}(n^{d+1})-1$ for some fixed $d$. We assume that the alphabet of the machine is the set of three
symbols $\{0, 1, \square\}$ where $\square$ denotes the blank symbol.
We decompose the tape into two disjoint
higher-order predicates. One predicate tracks the set of addresses where the
tape bit is set to $1$, while the other tracks addresses where the bit is set to
$0$. Any address not present in either set is implicitly treated as containing
the $\square$ symbol.

\subsection{Representation of Numbers}

Indexing an exponentially large tape requires a mechanism for representing
exponentially large natural numbers. We adopt the representation scheme
described by~\cite{iclp25} where predicates of order $k$ are used to represent
numbers of magnitude $\exp_{k}(n^d)$. As we will demonstrate later, this
representation constitutes a Stratified Linear Higher-Order Datalog program.

% \paragraph{Representing Polynomially-Big Numbers:}
Natural numbers up to $n^{d+1}-1$ are represented using tuples of individual constants
with a fixed length of $d+1$, where $d$ is an arbitrary but fixed natural number.
The following predicates define the ``first'' and ``last'' of such
numbers (denoted by \lstinline`first$_0$` and \lstinline`last$_0$`) and the
``less-than'' and ``successor'' relations on them (denoted by \lstinline`lt$_0$`
and \lstinline`succ$_0$`). Notice that the following definitions rely on the
\lstinline`first` and \lstinline`last` predicates defined in Example~\ref{example1}.
\begin{lstlisting}
first$_0$(Ord,X$_0$,...,X$_d$):-first(Ord,X$_0$),...,first(Ord,X$_d$).
last$_0$(Ord,X$_0$,...,X$_d$):-last(Ord,X$_0$),...,last(Ord,X$_d$).
lt$_0$(Ord,X$_0$,...,X$_d$,Y$_0$,...,Y$_d$):-Ord(X$_d$,Y$_d$).
lt$_0$(Ord,X$_0$,...,X$_d$,Y$_0$,...,Y$_d$):-Ord(X$_{d-1}$,Y$_{d-1}$),(X$_d$=Y$_d$).
$\ldots$
lt$_0$(Ord,X$_0$,...,X$_d$,Y$_0$,...,Y$_d$):-Ord(X$_0$,Y$_0$),(X$_1$=Y$_1$),...,(X$_d$=Y$_d$).
succ$_0$(Ord,$\bar{\tt X}$,$\bar{\tt Y}$):-lt$_0$(Ord,$\bar{\tt X}$,$\bar{\tt Y}$),not nsequential$_0$(Ord,$\bar{\tt X}$,$\bar{\tt Y}$).
nsequential$_0$(Ord,$\bar{\tt X}$,$\bar{\tt Y}$):-lt$_0$(Ord,$\bar{\tt X}$,$\bar{\tt Z}$),lt$_0$(Ord,$\bar{\tt Z}$,$\bar{\tt Y}$).
\end{lstlisting}

% \paragraph{Representing Exponentially-Big Numbers:}
% We now demonstrate how we can represent ``exponentially-big'' numbers as
% higher-order relations.
% In the foregoing discussion we will need the following
% notation: $\expk{0}(x) = x$ and $\expk{n+1}(x) = 2^{\expk{n}(x)}$.
We extend this system to represent ``exponentially-big'' numbers using higher-order relations.
Let $N_0 = n^{d+1} - 1$ be the largest number that can be represented by
$(d+1)$-tuples of individual constants and for $k \geq 1$, let $N_k$ be the
largest number that can be represented by using $k$-order relations. We can
exponentially increase the numbers up to the number $N_{k+1} = \exp_1(N_k +
1)-1$ by using $(k+1)$-order relations. Generally, it holds that $N_{k} =
\exp_{k}(n^{d+1})-1$.

If the $k$-order relations representing numbers up to $N_k$ are of type $\rho$,
then it suffices to use higher-order relations of type $\rho \rightarrow o$ in
order to represent numbers up to $N_{k+1}$. This is essentially a binary
representation where the lower order numbers denote bit positions. Formally, let
\lstinline`Z` be a $(k+1)$-order element and ${\tt R}_0,\ldots,{\tt R}_{N_k}$ be
the ordering of the elements that represent numbers in the previous counting
module. Let $f$ be the function mapping $\mtrue$ to $1$ and $\mfalse$ to $0$.
Then, we have $num({\tt Z}) = f({\tt Z(R}_{0})) + f({\tt Z(R}_{1})) \cdot 2 +
\cdots + f({\tt Z(R}_{N_k}))\cdot 2^{N_k}$. We begin with predicates testing for
the first and the last number.
\begin{lstlisting}
first$_{k+1}$(Ord,N):-not nfirst$_{k+1}$(Ord,N).
nfirst$_{k+1}$(Ord,N):-N(X).
last$_{k+1}$(Ord,N):-not nlast$_{k+1}$(Ord,N).
nlast$_{k+1}$(Ord,N):-not N(X).
\end{lstlisting}

The following definitions describe the ``less than'' relation between two
elements that represent numbers. We examine if a number \lstinline`N` is less
than \lstinline`M` by comparing the two numbers bit by bit in their binary
representation. The successor of a number is defined with the use of less-than.
\begin{lstlisting}
lt$_{k+1}$(Ord,N,M):-last$_{k}$(Ord,X),bit$_{k+1}$(Ord,N,M,X).
bit$_{k+1}$(Ord,N,M,X):-not N(X),M(X).
bit$_{k+1}$(Ord,N,M,X):-N(X),M(X),succ$_{k}$(Ord,Y,X),bit$_{k+1}$(Ord,N,M,Y).
bit$_{k+1}$(Ord,N,M,X):-not N(X),not M(X),succ$_{k}$(Ord,Y,X),bit$_{k+1}$(Ord,N,M,Y).
succ$_{k+1}$(Ord,N,M):-lt$_{k+1}$(Ord,N,M), not nsequential$_{k+1}$(Ord,N,M).
nsequential$_{k+1}$(Ord,N,M):-lt$_{k+1}$(Ord,N,Z),lt$_{k+1}$(Ord,Z,M).
\end{lstlisting}
For the case where $k=0$ the variables \lstinline`X` and \lstinline`Y`
in the code above must be substituted with the tuple vectors $\bar{\tt X}$ and $\bar{\tt Y}$.

It is straightforward to verify that the counting modules of any order $k\geq 0$ constitutes
a stratified linear program.
First, notice that there is a strict dependency on the order of the predicates:
predicates of order $k+1$ depend either on other predicates of order $k+1$ or on predicates of strictly lower order.
Consequently, all predicates of order $k$ can be assigned to strata strictly lower than those
of order $k+1$. Second, within the specific order $k+1$ we
impose a stratification function $S$ as:
\[S(\mathtt{succ}_{k+1})>S(\mathtt{nsequential}_{k+1})>S(\mathtt{lt}_{k+1})>S(\mathtt{bit}_{k+1})\]
Predicates such as \lstinline`first$_{k+1}$` and \lstinline`last$_{k+1}$` can be placed
at the lowest stratum of the $(k+1)$-order predicates. Finally, regarding linearity,
the only recursive definition at order $k+1$ is the \lstinline`bit$_{k+1}$`.
By inspection of its rules, we observe that it depends on exactly one
literal of itself, satisfying the linearity condition.

\subsection{The simulation of the Turing machine}

We now demonstrate how any query that belongs to  \EXPSPACE[(k-1)] ($k\geq 1$), can be
expressed by a stratified linear $(k+1)$-order Datalog$^\neg$ program. The simulation
is markedly different from that of~\cite{iclp25} which was easier since it could utilize
the full power of unrestricted Higher-Order Datalog$^\neg$.

% \paragraph{Encoding the input:}
%
Before presenting the simulation of the Turing machine $M$, we mention certain
simplifying assumptions, which do not affect the generality of the subsequent
results.
\begin{itemize}
  \item The input database consists of a single binary relation \lstinline|in|
        and the output database is also a single binary relation
        \lstinline|out|. In the following, the number of constants in the input
        database is denoted by $n$.
  \item The alphabet of $M$ that will be simulated is $\Sigma =\{0,1,\Box\}$.
        $M$ expects the input relation \lstinline|in| as the standard binary
        encoding of a graph, which is based on the ordering of the individual
        constants, in the first $n^2$ cells of its tape. For example, if the
        pair $(x,y)$ belongs to \lstinline|in|, then the tape of $M$ contains a
        ``1'' at cell position $num(x)+num(y) \cdot n$, otherwise it contains
        ``0''.
  \item $M$ decides whether a tuple $(a,b)$ belongs to the output relation
        \lstinline|out|. The next $n^2$ cells of its tape are used to encode
        $(a,b)$. All these cells contain the symbol ``0'', except for the cell
        at position $num(a)+num(b) \cdot n+n^2$ which contains ``1''.
  \item $M$ reaches its accepting state $yes$ if and only if the tuple $(a,b)$ belongs to the output relation
        \lstinline|out|.

\end{itemize}
We assume a Turing machine $M$ of three symbols $\Sigma =\{0,1,\Box\}$. By
convention, we use two predicates \lstinline|input$_{0}$| and
\lstinline|input$_{1}$| to encode the input tape of $M$. Specifically,
\lstinline|input$_{0}$| holds at position \lstinline|$\bar{\tt X}$| iff the
tape contains ``0'' at \lstinline|$\bar{\tt X}$|, \lstinline|input$_{1}$|
holds at \lstinline|$\bar{\tt X}$| iff the tape contains ``1'' at
\lstinline|$\bar{\tt X}$|, and the tape contains the blank symbol
``$\Box$'' iff neither holds. By construction below, the two
predicates are disjoint.
% Notice that since this encoding is priority based there are multiple states of
% these two predicates that correspond to the same tape, but that is not an issue
% for the simulation.

The following two predicates encode the binary input relation \lstinline|in|
and the tuple $(a,b)$
as a binary string.
In the following rules, we pad the binary input relation
with fixed constants to form $(d+1)$-arity tuples representing valid tape addresses.
\begin{lstlisting}
input$_{0}$(A,B,Ord,X,Y,Z$_2$,...,Z$_d$):-first(Ord,Z$_2$),...,first(Ord,Z$_d$),not in(X,Y).
input$_{0}$(A,B,Ord,Z$_0$,Z$_1$,Z$_2$,...,Z$_d$):-not(Z$_0$=A),first(Ord,Z),succ(Ord,Z,Z$_2$),
                               first(Ord,Z$_3$),...,first(Ord,Z$_d$).
input$_{0}$(A,B,Ord,Z$_0$,Z$_1$,Z$_2$,...,Z$_d$):-not(Z$_1$=B),first(Ord,Z),succ(Ord,Z,Z$_2$),
                               first(Ord,Z$_3$),...,first(Ord,Z$_d$).
input$_{1}$(A,B,Ord,X,Y,Z$_2$,...,Z$_d$):-first(Ord,Z$_2$),...,first(Ord,Z$_d$),in(X,Y).
input$_{1}$(A,B,Ord,Z$_0$,Z$_1$,Z$_2$,...,Z$_d$):-(Z$_0$=A),(Z$_1$=B),first(Ord,Z),succ(Ord,Z,Z$_2$),
                                first(Ord,Z$_3$),...,first(Ord,Z$_d$).
\end{lstlisting}

% \paragraph{Initial configuration of the Turing Machine}
% In order to represent the configurations of the Turing machine we use a
% higher-order predicate for each state.

We also need a higher-order predicate to lift the tuple representation of
numbers to the $(k-1)$-order numbering notation in order to transfer the input of
the tape from a zero order relation to a $(k-1)$-order relation. The predicate
\lstinline|lift$_{k-1}$(Ord,$\bar{\tt X}$,M)| transforms the number represented
by the tuple \lstinline|$\bar{\tt X}$| in the zero-order notation to the same
number \lstinline`M` in the $(k-1)$-order notation.
\begin{lstlisting}
lift$_{k-1}$(Ord,$\bar{\tt X}$,M):-first$_0$(Ord,$\bar{\tt X}$),first$_{k-1}$(Ord,M).
lift$_{k-1}$(Ord,$\bar{\tt X}$,M):-succ$_0$(Ord,$\bar{\tt Z}$,$\bar{\tt X}$),lift$_{k-1}$(Ord,$\bar{\tt Z}$,M'),succ$_{k-1}$(Ord,M',M).
\end{lstlisting}

We then define the two predicates \lstinline|symbol$_0$| and \lstinline|symbol$_1$|
that are the $(k-1)$-order analogues of the input predicates \lstinline|input$_0$| and
\lstinline|input$_1$|.
\begin{lstlisting}
symbol$_{0}$(A,B,Ord,P):-input$_{0}$(A,B,Ord,$\bar{\tt X}$),lift$_{k-1}$(Ord,$\bar{\tt X}$,P).
symbol$_{1}$(A,B,Ord,P):-input$_{1}$(A,B,Ord,$\bar{\tt X}$),lift$_{k-1}$(Ord,$\bar{\tt X}$,P).
\end{lstlisting}

% \paragraph{Encoding of Machine's Transitions}
With the input prepared, we model the Turing machine's execution as a
reachability analysis on the graph of configurations. We adopt a forward
approach: a configuration is valid if it leads to an accepting state.

Consider the transition $(s,\texttt{0})\longrightarrow (s',\texttt{1},\text{R})$ which reads:
``if the current state of the machine is $s$ and its cursor reads the symbol
``0'', then the machine changes its state to $s'$, it writes ``1'' on the
current cursor position onto the tape and then moves the cursor to the right.''
We encode this logic such that the validity of the current state $s$ depends on the validity of the successor state $s'$.
%
% We write the following rule to show that these two states are connected in the states-graph.
%
\begin{lstlisting}
state$_{s}$(Ord,Tp$_1$,Tp$_0$,P):-Tp$_0$(P),not Tp$_1$(P),succ$_{k-1}$(Ord,P,P'),
                       state$_{s'}$(Ord,flip$_1$(Tp$_1$,P),flip$_0$(Tp$_0$,P),P').
\end{lstlisting}
This rule asserts that if the machine successfully accepts starting from state
$s'$, it also accepts starting from state $s$.
The successor state $s'$ reflects the updated tape and the cursor moved to
the right or left by enforcing \lstinline`succ$_{k-1}$(Ord,P,P')` or
\lstinline`succ$_{k-1}$(Ord,P',P)` respectively.

The transition rules above rely on the ability to modify the tape.
Recall that the tape is encoded by the two disjoint relations
\lstinline`Tp$_1$` and \lstinline`Tp$_0$`, holding the positions where the
tape contains ``1'' and ``0'' respectively. Writing the symbol ``$b$'' at
position \lstinline`P` therefore involves updating \emph{both} relations:
\lstinline`P` is added to \lstinline`Tp$_b$` and removed from the other;
e.g., writing ``1'' at \lstinline`P` puts \lstinline`P` into
\lstinline`Tp$_1$` and removes it from \lstinline`Tp$_0$`.
We capture these two elementary set operations with the helper predicates
\lstinline`flip$_1$` and \lstinline`flip$_0$`, defined as follows:
\begin{lstlisting}
flip$_1$(Tp,P,X):-eq(P,X).
flip$_1$(Tp,P,X):-Tp(X).
flip$_0$(Tp,P,X):-Tp(X),neq(P,X).
\end{lstlisting}
Intuitively, the partial application \lstinline`flip$_1$(Tp,P)` denotes the
relation $\mathtt{Tp}\cup\{\mathtt{P}\}$, i.e., \lstinline`Tp` extended so
that \lstinline`P` belongs to it, while \lstinline`flip$_0$(Tp,P)` denotes
$\mathtt{Tp}\setminus\{\mathtt{P}\}$, i.e., \lstinline`Tp` with \lstinline`P`
removed. Writing the symbol ``1'' at position \lstinline`P` is then
accomplished by passing \lstinline`flip$_1$(Tp$_1$,P)` and
\lstinline`flip$_0$(Tp$_0$,P)` to the recursive call, as in the transition
rule above; writing ``0'' is symmetric, with the roles of the two relations
exchanged.
Both definitions rely on the auxiliary predicates \lstinline`eq` and \lstinline`neq`,
which check whether two relations (of the same type) are equal or not and are defined as:
\begin{lstlisting}
neq(R,Q):-R(X),not Q(X).
neq(R,Q):-not R(X),Q(X).
eq(R,Q):-not neq(R,Q).
\end{lstlisting}

% \paragraph{Acceptance and Query Execution}
The recursion defined in the transition rules must eventually terminate.
We ground the simulation by defining the accepting state ``yes'' as true
for any valid tape configuration:
\begin{lstlisting}
state$_{yes}$(Ord,Tp$_1$,Tp$_0$,P):-not intersect(Tp$_1$,Tp$_0$).
intersect(R,Q):-R(X),Q(X).
\end{lstlisting}

Finally, we start the simulation by querying the initial configuration.
The following query checks if the initial state $s_0$
(with the initial cursor position and input symbols)
can effectively ``reach'' the accepting state $\mathit{yes}$ via the rules defined above:
\begin{lstlisting}
query(A,B):-ordering(Ord),first$_{k-1}$(Ord,P),
            state$_{s_0}$(Ord,symbol$_1$(A,B,Ord),symbol$_0$(A,B,Ord),P).
\end{lstlisting}
This means that the \lstinline`query` is true if the \lstinline`state$_{s_0}$` is
true for the initial tape state and with the cursor being at the beginning of
the tape which in turn is true if this configuration is reachable by
one accepting configuration.

We conclude the construction by verifying that the simulation is
a stratified linear program. The core of the simulation lies in the transition rules
for the \lstinline`state$_s$` predicates which form a single stratum of mutually recursive definitions.
Crucially, each rule depends recursively on exactly one \lstinline`state$_s$` predicate,
satisfying the linearity constraint. All auxiliary predicates, such as
\lstinline`succ$_{k+1}$`, \lstinline`flip$_0$` and \lstinline`flip$_1$` are
defined without dependence on any \lstinline`state$_s$` predicate allowing them to be
placed on lower strata.

\section{Space-Efficient Computation of Stratified Linear Higher-Order Programs\label{computation}}
Using a simple-minded bottom-up approach to compute the answer to a query, would
require space tower-exponential in the order of the given program. To capture the precise
space complexity of Stratified Linear Higher-Order Datalog$^\neg$, we must avoid
producing and storing the entire model of a program.
In this section, we present a space-efficient computation that proceeds in a
top-down fashion. We demonstrate that for a $(k+1)$-order Stratified Linear Datalog$^\neg$
program, a query $p(\bar{d})$ can be evaluated using memory bounded by
$\EXPSPACE[(k-1)]$. We formalize this result in the following theorem:

\begin{theorem}\label{upper-bound-theorem}
Let $\Prog$ be a Stratified Linear $(k+1)$-Order Datalog$^\neg$ program
that defines a query. Then, there exists a deterministic Turing
machine that takes as input an encoding of a database $D$
that uses $n$ individual constant symbols and
a ground atom $p(\bar{a})$, where $p$ is a predicate constant of $\Prog$ and
$\bar{a}$ is a tuple of these individual constants, and
decides whether $p(\bar{a}) \in \mathcal{Q}_\mathsf{P}(D)$, while using at most
$\exp_{k-1}(n^d)$ tape cells, for some constant $d$.
\end{theorem}

%\subsection{Evaluation Strategy}
%
Our evaluation strategy relies on the stratified and linear nature of the program.
The meaning of a stratified $(k+1)$-order Datalog program can be computed stratum by stratum,
starting from the lowest and moving upwards.
To evaluate queries in the $(m+1)$-th stratum, we employ two routines:
\begin{itemize}
\item A query machine for the $(m+1)$-th stratum that decides the truth of an atom $(\mathsf{p}\,d_1\cdots d_n)$
      where $\mathsf{p}$ is a predicate belonging to the $(m+1)$-th stratum and the $d_i$ are
      \emph{representations of expression values} (see explanation later in the section).
      For the initial query, the $d_i$ are just individual constants.
\item An expression evaluator for the $m$-th stratum that, given a variable assignment,
      computes the value of a given expression $\mathsf{E}$ of a type with order at most $k$, with the
      restriction that $\mathsf{E}$ is formed using predicates solely from the first $m$ strata.
\end{itemize}
The query machine for the $(m+1)$-th stratum requires that the corresponding expression evaluator
for the $m$-th stratum is already defined since it invokes it. On the other hand, the evaluator routine
for the $m$-th stratum requires the query routine for any stratum $m' \leq m$
to be defined.
%Therefore, the construction assumes that we build each routine iteratively starting
%from the first stratum and moving upwards.
If we take the first stratum to be the  set of the (first-order) input predicates supplied from the
database in the form of facts, then for $m=0$ these routines are easily defined and thus omitted.

We now define the above routines at a high level of abstraction. A more formal presentation and the corresponding
correctness proofs, are given in the
\ifincludeappendix Appendix~\ref{appendix_computation} \else supplementary material \fi.

\vspace{0.2cm}

\noindent
{\bf Query Machine for the $(m+1)$-th Stratum}:
\begin{enumerate}
\item Start with the atom $(\mathsf{p}\,d_1\cdots d_n)$ written on the tape, where the $d_i$ are
      representations of expression values, \ie a string that describes the value (for example, a relation
      can be represented as a string listing all tuples in some predetermined order).
\item\label{select-step} Non-deterministically select a rule
      $\mathsf{p}\ \mathsf{R}_1 \cdots \mathsf{R}_n \lrule \mathsf{L}_1, \dots, \mathsf{L}_r$
      from $\Prog$ for predicate $\mathsf{p}$ and guess a valid assignment for the rule's variables such that
      $d_i$ is assigned to $\mathsf{R}_i$ for all $i \in \{ 1,\dots, n\}$.
\item Iterate through the literals $\mathsf{L}_i$ in the rule body:
\begin{itemize}
  \item If $\mathsf{L}_i$ is a negative literal then $\mathsf{L}_i$ is of the form $({\tt not}\ \mathsf{E})$.
        Invoke the expression evaluator for the $m$-th stratum, passing as parameter the guessed
        variable assignment, to compute $\mathsf{E}$.
        If the evaluation returns $\mtrue$, then reject, otherwise proceed to the next literal.
  \item If $\mathsf{L}_i$ is a positive literal and there is no predicate in $\mathsf{L}_i$ from the $(m+1)$-th stratum,
        then invoke the expression evaluator for the $m$-th stratum, passing as parameter the guessed
        variable assignment, to compute $\mathsf{L}_i$.
        If the evaluation returns $\mtrue$, then proceed to the next literal, otherwise reject.
  \item If there is a predicate in $\mathsf{L}_i$ from the $(m+1)$-th stratum, then this is the unique
        recursive literal $\mathsf{q}\ \mathsf{E}_1 \cdots \mathsf{E}_t$ implied by linearity.
        Every $\mathsf{E}_j$ is of order at most $k$ and includes predicates from
        the first $m$ strata. Iterate through $\mathsf{E}_1,\dots,\mathsf{E}_t$ invoking each time the evaluator
        for the $m$-th stratum, passing as parameter the guessed
        variable assignment, to compute each $  \mathsf{E}_j$ and store the representations $d'_1,\dots,d'_t$.
        Construct a new ground atom $\mathsf{q}\,d'_1 \cdots d'_t$.
\end{itemize}
\item If all literals $\mathsf{L}_i$ evaluate to $\mtrue$ without encountering a recursive literal,
      then accept.
\item If one of the literals $\mathsf{L}_i$ is a recursive one, overwrite the current goal with $\mathsf{q}\,d'_1 \cdots d'_t$
      and loop back to Step~\ref{select-step}.
\end{enumerate}
Note that the aforementioned query machine is non-deterministic and may contain
computational paths that do not terminate (\eg consider the rule
\lstinline`p:-p`). This machine is not technically a decider of whether
$(\mathsf{p}\,d_1\cdots d_n)$ is true or not. However, for machines that are
space bounded (more specifically by a space constructible function
$f(n)$~\citep{Sipser}) there exists a corresponding non-deterministic Turing
machine where every computational path terminates. By~\cite{Savitch1970}, there
also exists a deterministic machine that computes the same query within the same
space complexity class. In the following when we refer to the query machine we
mean the deterministic version.

We now define the expression evaluator for the $m$-th stratum. Let
$\mathsf{E}$ be an expression of order at most $k$ that includes predicates from
the first $m$ strata. The expression evaluator makes calls to query machines in
order to construct a representation of the expression's value.

%This is achieved by successively calling the query machines to fully construct a canonical
%representation of the expression's value.
%(Notice that since $\mathsf{E}$ is an expression of
%$k+1$-order Datalog each constant in it has an order of at most $k+1$ and each variable of at most $k$.)

\vspace{0.2cm}

\noindent
{\bf Expression Evaluator for the $m$-th Stratum}:
\begin{itemize}
\item If the expression $\mathsf{E}$ is an individual constant,  then return the constant.
      %the evaluation is immediate:
      %the result is simply the constant itself stored in our predetermined encoding for the $\iota$ types in memory.
\item If the expression $\mathsf{E}$ is of the form $(\mathsf{q}\ \mathsf{E}_1\cdots\mathsf{E}_t)$
      where $\mathsf{q}$ belongs to the $m'$-th stratum, $m' \leq m$, and its type is
      $\rho_1 \to \dots \to \rho_t \to  \cdots \to \rho_{t'} \to o$. Then:
      \begin{enumerate}
        \item Evaluate the expressions $\mathsf{E}_1, \dots, \mathsf{E}_t$ sequentially
              to obtain and store their representations $d_1,\dots,d_t$. This is performed
              by calling the evaluator recursively for each subexpression.
        \item To determine the representation of the result, the machine enumerates all possible
              tuples  $d_{t+1}, \dots, d_{t'}$ of representations regarding the remaining arguments.
              For each tuple, it invokes the query machine for the $m'$-th stratum with the goal
              $(\mathsf{q}\, d_1 \cdots d_{t}\, d_{t+1} \cdots d_{t'})$.
      \end{enumerate}
\item If the expression $\mathsf{E}$ is of the form $(\mathsf{X}\ \mathsf{E}_1\cdots\mathsf{E}_t)$
      where $\mathsf{X}$ is a variable (of order at most $k$) then it evaluates the expressions $\mathsf{E}_1, \dots, \mathsf{E}_t$
      to obtain and store the representation of their values $d_1,\dots,d_t$, and then uses
      the variable assignment of $\mathsf{X}$ to compute the representation of the value of
      the application.
\end{itemize}

The last two bullets implicitly handle the cases where there are zero arguments.

\vspace{0.2cm}

We now briefly comment on the space complexity; the full and precise arguments are given in the
\ifincludeappendix Appendix~\ref{appendix_computation} \else supplementary material \fi.
The query machine's space complexity depends on storing the current goal, guessing variable assignments, and evaluating lower-strata subroutines. For a $(k+1)$-order program, arguments are at most order $k$, belonging to a domain of size $\exp_k(poly(n))$. Representing these arguments requires unique identifiers of logarithmic size, \ie $\exp_{k-1}(poly(n))$ bits. Therefore, storing the current goal and variable assignments takes $\exp_{k-1}(poly(n))$ space. Crucially, program linearity allows the machine to overwrite the current goal during recursion, preventing depth-related memory blow-up. Furthermore, lower-strata expressions are evaluated sequentially (reusing memory), and since the number of strata is a fixed constant independent of the input size $n$, the call stack depth is strictly bounded. Consequently, the overall space complexity remains $\exp_{k-1}(poly(n))$.

\section{Future Work: Implementation of 2nd-Order Stratified Linear Datalog$^\neg$\label{implementation}}
The space complexity results described in this paper are not merely theoretical;
they provide the foundation for a practical implementation of the 2nd order fragment
of Stratified Linear Datalog$^\neg$. By capturing \textsf{PSPACE}, this fragment addresses
many problems that remain out of reach for traditional ASP systems.

The key idea behind our implementation is to compile any given 2nd order Stratified Linear Datalog$^\neg$
program and user query into a Quantified Boolean Formula (QBF). We can then use existing, efficient QBF
solvers to evaluate the resulting formula and return answers to the initial query. Because evaluating QBFs
is the canonical \textsf{PSPACE}-complete problem~\citep{Sipser}, this translation offers a direct and highly
feasible approach to implementing our source language.

Due to space limitations, the remainder of this section provides only a high-level overview of this
translation, leaving the formal generalization of this technique to a future paper.
\ifincludeappendix
To better understand the mechanics of this approach, interested readers should consult Appendix~\ref{recoloring-qbf-eval},
which describes the steps of transforming a well-known \textsf{PSPACE}-complete problem written in 2nd
order Stratified Linear Datalog$^\neg$, into QBF.
\fi

The space-efficient computation described in Section~\ref{computation} models evaluation as a
top-down procedure that avoids materializing the full program interpretation. We
can conceptualize this evaluation as a reachability problem within a
configuration graph. In the first-order setting, linearity implies that the
evaluation engine only needs to store the current state of the computation
(which is a classical Datalog atom). In the 2nd-order case, the active computation
is, in general, a first-order relation. For an input database with a domain of size $n$, a
first-order relation of arity $d$ can be represented by a truth assignment to $O(n^d)$ boolean
variables. Therefore, a single recursive step in a 2nd-order
Stratified Linear Datalog$^\neg$ program represents a transition from one boolean
configuration to another. Because the program is linear, meaning there exists at
most one predicate constant in the body from the same stratum, the core query
evaluation simply asks whether there is a valid logical path from an initial
configuration to an accepting configuration.

The transformation to QBF relies on encoding these configurations as boolean
vectors and capturing the linear recursion using a reachability formula. Let a
variable $\mathsf{X}$ representing a first-order relation be encoded as a vector of
boolean variables $\vec{x}$ of length $n^d$. The linear recursive rules of the current stratum,
along with any non-recursive rules or sequences of rules from strictly lower strata they depend on,
can be compiled into a quantifier-free boolean formula $T(\vec{x}, \vec{y})$.
This formula evaluates to true if and
only if state $\vec{y}$ can be derived from state $\vec{x}$ in exactly one step.
Because the state space of possible configurations has a size of $2^{n^d}$, a linear
query might require up to $2^{n^d}$ recursive steps to reach the target state. Unrolling
the transition formula $T(\vec{x}, \vec{y})$ exponentially would require exponential space,
which contradicts our \textsf{PSPACE} bounds. Fortunately, we can follow the principles
of Savitch's theorem and use the expressive power of universal quantifiers in
QBF to compress the path search.
\ifincludeappendix
A detailed explanation of how Savitch's theorem is used for this compression, is given in Appendix~\ref{recoloring-qbf-eval}.
\fi

This translation mechanism offers a practical evaluation pipeline. A compiler
can map all first-order variables to boolean vectors, compile the linear
recursive rules into the base transition formula $T$, and generate appropriate
quantified reachability formulas based on Savitch's theorem. The resulting formula
can then be evaluated by modern QBF solvers. By avoiding the
materialization of the full interpretation, this approach strictly adheres to
the space-efficient nature of our fragment and provides a viable
implementation pathway for 2nd order Stratified Linear Datalog$^\neg$ programs.

As another promising direction for future work, it is worth investigating whether we can relax the stratification
constraints by introducing a notion of local stratification for Higher-Order
Datalog$^\neg$. Although our current space complexity arguments rely on a fixed number
of strata, we conjecture that if the number of strata are polynomial with
respect to the input size, the resulting linear programs would still preserve
their space efficiency.
% \begin{itemize}
% \item Stratified linear programs and stratified programs in general have a unique two-valued model.
%       One point of direction is to investigate stratified linear programs with choices.
%       We conjecture that the complexity will not change since $\EXPSPACE[k]=k-\mathsf{NEXPSPACE}$ for $k \geq 1$.
% \end{itemize}

%\vspace{-0.3cm}

\bibliography{expressivity}

\ifincludeappendix
\clearpage
\appendix
% \section{The Semantics of Higher-Order Datalog$^\neg$}
\section{Proofs of Section~\ref{computation}}
\label{semantics}\label{appendix_computation}

In this appendix, we supply the technical details omitted from the main text
regarding the evaluation of Stratified Linear Higher-Order Datalog$^\neg$. 
We define the precise semantics of the language and provide the
proofs for the upper bound complexity characterized in Theorem~\ref{upper-bound-theorem}.

The semantics of the base type $\bool$ is the classical Boolean domain
$\{\mathit{true}, \mathit{false}\}$ and that of the base type $\basedom$ is
$U_{\mathsf{P}}$, namely the set of individual constant symbols in \Prog.
% The semantics of types of the form $\rho \to \pi$ is the set
% of all functions from the domain of type $\rho$ to that of type $\pi$. 
% We define, simultaneously with the meaning of every type, a partial order on the
% elements of the type.
We start by defining the meanings of types of Higher-Order Datalog${^\neg}$.

\begin{definition}\label{def:orders_two-valued}
Let $\mathsf{P}$ be a Higher-Order Datalog${^\neg}$ program. We define the meaning $\mo{\tau}_{U_{\mathsf{P}}}$ 
of a type $\tau$ with respect to $U_{\mathsf{P}}$, as follows:
\begin{itemize}
  \item $\mo{\bool}_{U_{\mathsf{P}}} = \{\mathit{true}, \mathit{false}\}$.
  The partial order $\leq_\bool$ is the one induced by $\mfalse <_\bool \mtrue$.

  \item $\mo{\basedom}_{U_{\mathsf{P}}} = U_{\mathsf{P}}$.
        The partial order $\leq_\basedom$ is the
        trivial one defined as $d \leq_\basedom d$ for all $d \in U_{\mathsf{P}}$.
  \item $\mo{\rho \to \pi}_{U_{\mathsf{P}}} = \mo{\rho}_{U_{\mathsf{P}}} \to \mo{\pi}_{U_{\mathsf{P}}}$,
        namely the set of all functions from  $\mo{\rho}_{U_{\mathsf{P}}}$ to $\mo{\pi}_{U_{\mathsf{P}}}$.
        The partial order $\leq_{\rho \to \pi}$ is defined as follows:
        for all $f,g \in \mo{\rho \to \pi}_{U_{\mathsf{P}}}$,
        $f \leq_{\rho \to \pi} g$ iff $f(d) \leq_{\pi} g(d)$ for all $d \in \mo{\rho}_{U_{\mathsf{P}}}$.
\end{itemize}
\end{definition}

The subscripts from the above partial order will be omitted when they are
obvious from context. Moreover, we will omit the subscript $U_{\mathsf{P}}$
assuming that our semantics is defined with respect to a specific program
$\mathsf{P}$.
For every predicate type $\pi$, $(\mo{\pi}, \leq_\pi)$ is a complete lattice. 
We denote by $\bigvee_{\leq_{\pi}}$ and
$\bigwedge_{\leq_{\pi}}$ the corresponding lub and glb operations of the above
lattice.

To denote that an expression $\mathsf{E}$ has type $\rho$ we will often write $\mathsf{E}:\rho$.

\begin{definition}\label{def:interpretation_Herbrand}
A Herbrand interpretation $I$
of a program $\mathsf{P}$ assigns to each
individual constant $\mathsf{c}$ of $\mathsf{P}$, the element
$I(\mathsf{c}) = \mathsf{c}$, and to each predicate constant
$\mathsf{p} : \pi$ of $\mathsf{P}$, an element $I(\mathsf{p}) \in \mo{\pi}$.
% An interpretation is called two-valued if for every $\mathsf{p} : \pi$ of $\mathsf{P}$, $I(\mathsf{p}) \in \mo{\pi}$.
%
\end{definition}

We will denote the set of Herbrand interpretations of a program $\mathsf{P}$
with $H_\mathsf{P}$.
We define a partial order on $H_\mathsf{P}$ as
follows: for all $I, J \in H_\mathsf{P}$, $I \leq J$
iff for every predicate constant $\mathsf{p} : \pi$ that appears in
$\mathsf{P}$, $I(\mathsf{p}) \aleq[\pi] J(\mathsf{p})$.

\begin{definition}\label{def:state_Herbrand}
A \emph{Herbrand state} $s$ of a program $\mathsf{P}$ is a function that
assigns to each argument variable $\mathsf{R}$ of type $\rho$, an element
$s(\mathsf{R}) \in \mo{\rho}$.  We denote the set of Herbrand states with
$S_\mathsf{P}$.
\end{definition}
In the following, $s[\mathsf{R}_1/d_1,\ldots,\mathsf{R}_n/d_n]$ is used to
denote a state that is identical to $s$ the only difference being that the new
state assigns to each $\mathsf{R}_i$ the corresponding value $d_i$; for brevity,
we will also denote it by $s[\bar{\mathsf{R}}/\bar{d}]$.

\begin{definition}\label{def-semantics}
Let $\mathsf{P}$ be a Higher-Order Datalog${^\neg}$ program, $I$ a Herbrand
interpretation of $\mathsf{P}$, and $s$ a Herbrand state.
The semantics of expressions, literals and bodies is defined as follows:
\begin{enumerate}
  \item $\mwrs{\mathsf{R}}{I}{s} = s(\mathsf{R})$
  \item $\mwrs{\mathsf{c}}{I}{s} = I(\mathsf{c}) = \mathsf{c}$
  \item $\mwrs{\mathsf{p}}{I}{s} = I(\mathsf{p})$
  \item $\mwrs{(\mathsf{E}_1\ \mathsf{E}_2)}{I}{s} = \mwrs{\mathsf{E}_1}{I}{s}\ \mwrs{\mathsf{E}_2}{I}{s}$
  \item $\mwrs{(\mathsf{E}_1\approx \mathsf{E}_2)}{I}{s} = \begin{cases}
    \mathit{true},  & \text{if } \mwrs{\mathsf{E}_1}{I}{s} = \mwrs{\mathsf{E}_2}{I}{s} \\
    \mathit{false}, & \text{otherwise}
    \end{cases}$
  \item $\mwrs{(\mathsf{not}\ \mathsf{E})}{I}{s} = (\mwrs{\mathsf{E}}{I}{s})^{-1}$, with $\mathit{true}^{-1}\!=\!\mathit{false}$, $\mathit{false}^{-1}\!=\!\mathit{true}$
  \item $\mwrs{(\mathsf{E}_1 \wedge \cdots \wedge \mathsf{E}_m)}{I}{s} =
    \bigwedge_{\leq_\bool}\{\mwrs{\mathsf{E}_1}{I}{s},\ldots,\mwrs{\mathsf{E}_m}{I}{s}\}$
\end{enumerate}
\end{definition}

\begin{definition}\label{def:model}
Let $\mathsf{P}$ be a program and $M$ be a Herbrand
interpretation of $\mathsf{P}$. Then, $M$ is a
\emph{Herbrand model} of $\mathsf{P}$ iff for every rule
$\mathsf{p}\ \overline{\mathsf{R}} \lrule \mathsf{B}$ in $\mathsf{P}$
and for every Herbrand state $s$,
$\mwrs{\mathsf{B}}{M}{s} \leq_\bool \mwrs{\mathsf{p}\ \bar{\mathsf{R}}}{M}{s}$.
\end{definition}

% When viewing elements of $\mos{\pi}$ as partial sets and elements of $\mo{\pi}$
% as sets, the isomorphism maps a partial set onto the pair with first
% component all the \emph{certain} elements of the partial set (those mapped to
% $\mtrue$) and second component all the \emph{possible} elements (those mapped
% to $\mtrue$ or $\mundef$).

%
Let $(L,\leq)$ be a complete lattice. We define $L^{c} =\{(x,y) \in L \times L \mid x \leq y\}$.
Moreover, we define the relations $\leq$ and  $\preceq$ on $L^{c}$, so that
for all $(x,y),(x',y') \in L^c$:
%
% \begin{itemize}
$(x, y) \leq (x', y')$ iff $x \leq x'$ and $y \leq y'$, and
$(x, y) \preceq (x', y')$ iff $x \leq x'$ and $y' \leq y$.
% \end{itemize}
%

We will denote the \emph{first} and \emph{second} selection
functions on pairs with the more compact notation $[\cdot]_1$ and $[\cdot]_2$:
given any pair $(x,y)$, it is $[(x,y)]_1 = x$ and $[(x,y)]_2 = y$.
It is easy to see that $L^{c}$ is a complete lattice with respect to $\leq$ where $\bigvee_{\leq}$
and $\bigwedge_{\leq}$ are defined in a pointwise way and $L^{c}$ is a meet-semilattice with respect to $\preceq$
where $\bigwedge_{\preceq} S = (\bigwedge\{ [x]_1 \mid x \in S \}, \bigvee \{ [x]_2 \mid x \in S \})$.

The following definition defines an alternative but equivalent semantics of the expressions, literals and bodies
with respect to a pair of interpretations. This is necessary for proceeding to
the definition of the approximating operator and more convenient when doing proofs. This operator is 
$\prec$-monotonic.

\begin{definition}\label{def:pair-semantics}
Let $\mathsf{P}$ be a program, $(I, J) \in H^c_\mathsf{P}$, and $s \in S_{\mathsf{P}}$.
The pair semantics of expressions, literals and bodies is defined as follows:
\begin{enumerate}
  \item $\mwrc{\mathsf{R}}{I,J}{s} = (s(\mathsf{R}),s(\mathsf{R}))$
  \item $\mwrc{\mathsf{c}}{I,J}{s} = (I(\mathsf{c}), J(\mathsf{c}))$
  \item $\mwrc{\mathsf{p}}{I,J}{s} = (I(\mathsf{p}), J(\mathsf{p}))$
  \item $\mwrc{(\mathsf{E}_1\ \mathsf{E}_2)}{I,J}{s} = (\bigwedge_{\leq_\pi} \{f(d) \mid d \in \lsem \rho\rsem, l \leq d \leq u \}, \bigvee_{\leq_\pi} \{g(d) \mid d \in \lsem \rho\rsem, l \leq d \leq u\})$,
      where $(f,g) = \mwrc{\mathsf{E}_1}{I,J}{s}$, $(l,u) = \mwrc{\mathsf{E}_2}{I,J}{s}$ for $\mathsf{E}_1\! :\! \rho \to \pi$ and $\mathsf{E}_2\! :\! \rho$.
  \item $\mwrc{(\mathsf{E}_1\approx \mathsf{E}_2)}{I,J}{s} = \begin{cases}
    (\mathit{true}, \mathit{true}),  & \text{if } \mwrc{\mathsf{E}_1}{I,J}{s} = \mwrc{\mathsf{E}_2}{I,J}{s} \\
    (\mathit{false}, \mathit{false}), & \text{otherwise}
    \end{cases}$
  \item $\mwrc{(\mathsf{not}\ \mathsf{E})}{I,J}{s} =  \mwrc{\mathsf{E}}{I,J}{s}^{-1}$, with $(\mathit{true},\mathit{true})^{-1}\!=\!(\mathit{false},\mathit{false})$, $(\mathit{false},\mathit{false})^{-1}\!=\!(\mathit{true},\mathit{true})$
  and $(\mathit{false},\mathit{true})^{-1}\!=\!(\mathit{false},\mathit{true})$
  \item $\mwrc{(\mathsf{E}_1 \wedge \cdots \wedge \mathsf{E}_m)}{I,J}{s} =
  \bigwedge_{\leq_\bool}\{\mwrc{\mathsf{E}_1}{I,J}{s},\ldots,\mwrc{\mathsf{E}_m}{I,J}{s}\}$
\end{enumerate}
\end{definition}

\begin{definition}\label{def:pair-ap}
Let $\mathsf{P}$ be a program. The \emph{approximating operator} (or approximator)
$\ATP : H^c_\mathsf{P} \to H^c_\mathsf{P}$ is defined
for every predicate constant $\mathsf{p} : \rho_1 \to \cdots \to \rho_n \to \bool$ in $\mathsf{P}$ and
all $d_1 \in \mo{\rho_1},\ldots, d_n \in \mo{\rho_n}$, as $\ATP(I,J) = (\ATP(I,J)_1, \ATP(I,J)_2)$
where, for $i \in \{ 1, 2 \}$:
\[ \ATP(I,J)_i(\mathsf{p}) (\bar{d})= \bigvee\nolimits_{\leq_\bool}\{
  [\mwrc{\mathsf{B}}{I,J}{s[\bar{\mathsf{R}}/\bar{d}]}]_i \mid
          \mbox{$s\in S_{\mathsf{P}}$ and
                $(\mathsf{p}\ \bar{\mathsf{R}} \lrule \mathsf{B})$ in $\mathsf{P}$}\}
\]
\end{definition}

Several fixpoints of the approximating operator are of interest~\citep{iclp24}.
In this section, we focus on the stable fixpoints of $\ATP$: a pair $(I,J)$
is a stable fixpoint of $\ATP$ if and only if $I = \lfp \ATP(\cdot,J)_1$ and $J = \lfp \ATP(I,\cdot)_2$.

\begin{definition}\label{def:AFTsemantics}
Let $\mathsf{P}$ be a Higher-Order Datalog${^\neg}$ program.
% and  $\ATP(I,J) = \tau(\mathcal{T}_\mathsf{P}(\tau^{-1}(I,J)))$.
We call $I$ a \emph{stable model} of $\mathsf{P}$ if $(I,I)$ is a stable fixpoint of $\ATP$.
\end{definition}

We can decompose the program into subprograms and study the stable models of
each subprogram. The following definition formalizes this decomposition through
the notion of a splitting set, which identifies a sub-program that can be
evaluated independently of the remaining rules.

% \section{Proofs of Section~\ref{computation}}
% \label{appendix_computation}

\begin{definition}\label{def_splitting}
Let $\mathsf{P}$ be a program and $U$ a set of predicate constants.
We say that $U$ is a splitting set of $\mathsf{P}$
if for every clause $C$ of $\mathsf{P}$,
if a predicate constant of $U$ appears in the head of $C$, then every predicate constant appearing in $C$ is included in $U$.
The set of clauses $C\in \mathsf{P}$ such that all predicate constants appearing in $C$ are included in $U$
is called the \emph{bottom} of $\mathsf{P}$ relative to $U$ and denoted by $b_U(\mathsf{P})$.
The set $\mathsf{P} \setminus b_U(\mathsf{P})$ is called the \emph{top} of $\mathsf{P}$ relative to $U$.
\end{definition}

\begin{lemma}\label{splitting-set-models} Let $\mathsf{P}$ be a program and $U$
be a splitting set of $\mathsf{P}$ such that $b_U(\mathsf{P})$ contains all
constants of type $\iota$ of $\mathsf{P}$. If $M$ is a stable model of
$\mathsf{P}$, then $M$ restricted to $U$ is a stable model of
$b_U(\mathsf{P})$.
\end{lemma}
\begin{proof}
% Let $\ATP$ be the approximator of $\Prog$ and $(I,J)$ a pair interpretation for \Prog. 
% For any predicate $\mathsf{p}$ in $\Prog$ we have
% \[ \ATP(I,J)_i(\mathsf{p}) (\bar{d})= \bigvee\nolimits_{\leq_\bool}\{
%   [\mwrc{\mathsf{B}}{I,J}{s[\bar{\mathsf{R}}/\bar{d}]}]_i \mid
%   \mbox{$s\in S_{\mathsf{P}}$ and
%     $(\mathsf{p}\ \bar{\mathsf{R}} \lrule \mathsf{B})$ in $\mathsf{P}$}\}
% \]
% %
% Similarly, let $A_{b_U(\Prog)}$ be the approximator of $b_U(\Prog)$. For
% $(I_U,J_U)$ a pair interpretation for $b_U(\Prog)$ which is defined for any predicate
% $\mathsf{p}$ in $U$ as
% \[ A_{b_U(\Prog)}(I_U,J_U)_i(\mathsf{p}) (\bar{d})= \bigvee\nolimits_{\leq_\bool}\{
%   [\mwrc{\mathsf{B}}{I_U,J_U}{s[\bar{\mathsf{R}}/\bar{d}]}]_i \mid
%   \mbox{$s\in S_{b_U(\Prog)}$ and
%     $(\mathsf{p}\ \bar{\mathsf{R}} \lrule \mathsf{B})$ in ${b_U(\Prog)}$}\}
% \]
%
Let the pair $(I,J)$ be a stable fixpoint for $\Prog$ and let
$(I_U,J_U)$ be its restriction to the bottom program ${b_U(\Prog)}$.
We will show that $(I_U,J_U)$ is a stable fixpoint of ${b_U(\Prog)}$ which directly implies the statement.
It is easy to verify that $(I_U,J_U)$ is a fixpoint of the approximator $A_{b_U(\Prog)}$. 
Therefore, $A_{b_U(\Prog)}(X,J_U)_1$ and $A_{b_U(\Prog)}(I_U,Y)_2$ are well-defined
operators in the corresponding intervals. To see this is true, take for example 
any  interpretation $X \leq J_U$. Then we
have that $A_{b_U(\Prog)}(X,J_U)_1 \leq A_{b_U(\Prog)}(J_U,J_U)_1 \leq
A_{b_U(\Prog)}(J_U,J_U)_2 \leq A_{b_U(\Prog)}(I_U,J_U)_2 = J_U$. Similarly, for $I_U \leq Y$.

For the sake of contradiction assume that $(I_U,J_U)$ is not a stable fixpoint of 
${b_U(\Prog)}$. Then, let $I'_U = \lfp A_{b_U(\Prog)}(\cdot,J_U)_1$ and so $I'_U < I_U$.
We consider an interpretation $I'$ defined as: 
\[
  I'(\mathsf{p}) = \begin{cases} 
      I'_U(\mathsf{p}) & \text{if $\mathsf{p} \in U$} \\
      I(\mathsf{p}) & \text{otherwise}
  \end{cases}
\]
It is easy to verify that $I' < I$.
We will show that $I'$ is a prefixpoint of $\ATP(\cdot,J)_1$ which is a contradiction
since $I$ is the least fixpoint (and therefore its least prefixpoint). 
Indeed, for any $\mathsf{p}$ in $U$ we have that:
\begin{align*}
  \ATP(I',J)_1(\mathsf{p}) (\bar{d})
    & = \bigvee\nolimits_{\leq_\bool}\{
  [\mwrc{\mathsf{B}}{I',J}{s[\bar{\mathsf{R}}/\bar{d}]}]_1 \mid \mbox{$s\in S_{\mathsf{P}}$ and
  $(\mathsf{p}\ \bar{\mathsf{R}} \lrule \mathsf{B})$ in $\mathsf{P}$}\}                                      \\
    & = \bigvee\nolimits_{\leq_\bool}\{
  [\mwrc{\mathsf{B}}{I'_U,J_U}{s_u[\bar{\mathsf{R}}/\bar{d}]}]_1 \mid \mbox{$s_u\in S_{b_U(\Prog)}$ and
  $(\mathsf{p}\ \bar{\mathsf{R}} \lrule \mathsf{B})$ in ${b_U(\Prog)}$}\}                                    \\
    & = A_{b_U(\Prog)}(I'_U,J_U)_1(\mathsf{p}) (\bar{d})= I'_U(\mathsf{p}) (\bar{d}) = I'(\mathsf{p}) (\bar{d})
\end{align*}
This follows from the fact that all rules defining the predicate $\mathsf{p}$
are in $b_U(\mathsf{P})$ and that the bodies of these rules
involve only predicates interpreted by $(I_U, J_U)$.
% Furthermore, we rely on the fact that the Herbrand universe is the same in $\mathsf{P}$ and 
% $b_U(\mathsf{P})$ and {\bf that every valid state
% of $\Prog$ restricted is a valid state of ${b_U(\Prog)}$ and every state $s_U$ of
% ${b_U(\Prog)}$ when extended is a state of $\Prog$.}

For any $\mathsf{p}$ not in $U$ it follows by the monotonicity of $\ATP$ that 
$\ATP(I',J)_1(\mathsf{p}) \leq \ATP(I,J)_1(\mathsf{p})=I(\mathsf{p})$. Furthermore, 
since $I(\mathsf{p}) = I'(\mathsf{p})$
it follows $\ATP(I',J)_1(\mathsf{p}) \leq I'(\mathsf{p})$.

The case for $J_U$ is analogous.
\end{proof}

It is easy to see that for any program $\mathsf{P}$ that is stratified with a
stratification function $S$, for every natural number $n$, 
the set $U_n = \{ \mathsf{p} \mid S(\mathsf{p}) < n \}$ is a splitting set of $\mathsf{P}$.
Using these splitting sets, we can divide $\mathsf{P}$ into a finite number of
disjoint subprograms $\mathsf{P}_1, \ldots, \mathsf{P}_m$. In the rest of the
section we will denote a stratified program $\mathsf{P}$ as 
$\mathsf{P}_1 \cup \dots \cup \mathsf{P}_m$.
A key property of such programs is that they possess a unique stable model.

\begin{lemma}\label{stratified-has-unique-stable}
Let $\mathsf{P}$ be a stratified Higher-Order Datalog$^{\neg}$ program. Then, it has a unique stable model.
\end{lemma}
\begin{proof}
An immediate consequence of Theorems~7.3 and~7.4 from~\cite{iclp24}.
\end{proof}

In the following, whenever we decompose a program $\mathsf{P}$ into $\mathsf{P}_1 \cup \dots \cup \mathsf{P}_m$,
we will silently assume that every $\iota$ type constant appearing in $\mathsf{P}$ appears also in $\mathsf{P}_1 $ and therefore 
in every $\mathsf{P}_1 \cup \dots \cup \mathsf{P}_i$ for any $i \leq m$. 
This ensures that the Herbrand universe of every subprogram is the same as the Herbrand universe of $\mathsf{P}$ and so all types evaluate
from the same sets in each subprogram.
\begin{definition}\label{reduct}
Let $\mathsf{P} = \mathsf{P}_1 \cup \dots \cup \mathsf{P}_m$ 
be a stratified Higher-Order Datalog$^{\neg}$ program, $U_m$ be the splitting set of the predicates 
defined in $\mathsf{P}_1 \cup \dots \cup \mathsf{P}_{m-1}$ 
and $M_{m-1}$ be a Herbrand interpretation of these predicates. 
The propositional program $\mathsf{P}_m^{M_{m-1}}$ called the reduct of $\mathsf{P}_m$ with respect to $M_{m-1}$ 
is defined as the set of all propositional rules 
$\mathsf{p}_{s(\bar{\mathsf{R}})} \lrule \mathsf{q}^1_{\bar{d_1}},\dots,\mathsf{q}^k_{\bar{d_k}}$
constructed for every rule 
$\mathsf{p}\ \bar{\mathsf{R}} \lrule \mathsf{L}_1,\ldots,\mathsf{L}_n$
in  $\mathsf{P}_m$ and every Herbrand state $s$ that satisfies the following conditions:
\begin{enumerate}
\item Every literal $\mathsf{L}_i$ in the body whose predicates
      are in $U_m$, $\mwrs{\mathsf{L}_i}{M_{m-1}}{s} = \mtrue$.
\item The body atoms $\mathsf{q}^i_{\bar{d_i}}$ correspond exactly to the literals in the body 
      of the form $\mathsf{q}^i\ \mathsf{E}_1\dots\mathsf{E}_t$ where $\mathsf{q}^i$ is not in $U_m$. 
      Each literal maps to the propositional atom $\mathsf{q}^i_{\bar{d}_i}$ where $\bar{d}_i$ is the tuple 
      $\langle \mwrs{\mathsf{E_1}}{M_{m-1}}{s}, \ldots, \mwrs{\mathsf{E_t}}{M_{m-1}}{s} \rangle$.
\end{enumerate}
\end{definition}

%Note that because the original program $\Prog_m$ is stratified linear, the reduct $\mathsf{G}_m$ is also linear.
The following lemma follows easily by the definition, so its proof is omitted.

\begin{lemma}
Let $\mathsf{P} = \mathsf{P}_1 \cup \cdots \cup \mathsf{P}_m$ be a stratified 
Higher-Order Datalog$^{\neg}$ program and $M_{m}$ the unique stable model of $\mathsf{P}$.
Moreover, let $M_{m-1}$ be the unique stable model of $\mathsf{P}_1 \cup \cdots \cup \mathsf{P}_{m-1}$.
If $\mathsf{P} $ is stratified linear then $\mathsf{P}^{M_{m-1}}_m$ is also stratified linear.
More specifically, $\mathsf{P}^{M_{m-1}}_m$ contains at most one atom in each body of a rule.
\end{lemma}
% \begin{proof}
% Easy to verify from Definition~\ref{reduct}.
% \end{proof}

We now establish the correspondence between the model of a stratified program $\mathsf{P}$ and 
those of its reducts with respect to the models of the previous subprograms.

\begin{lemma}\label{reduct-models}
Let $\mathsf{P} = \mathsf{P}_1 \cup \cdots \cup \mathsf{P}_m$ be a stratified 
Higher-Order Datalog$^{\neg}$ program and $M_{m}$ the unique stable model of $\mathsf{P}_1 \cup \cdots \cup \mathsf{P}_m$
%Furthermore, assume that all $\iota$ 
%constants that appear in $\mathsf{P}$ also appear in $\mathsf{P}_1$.
, $M_{m-1}$ be the unique stable model of $\mathsf{P}_1 \cup \cdots \cup \mathsf{P}_{m-1}$ and 
$N_m$ be the \emph{minimum} model of $\mathsf{P}^{M_{m-1}}_m$. Then, for every 
propositional atom $\mathsf{p}_{\bar{d}}$ in $\mathsf{P}^{M_{m-1}}_m$, $N_m(\mathsf{p}_{\bar{d}}) = \mtrue$ 
iff $M_m(\mathsf{p})(\bar{d}) = \mtrue$.
\end{lemma}
\begin{proof}
Let $U_{m}$ be the splitting set containing all predicate constants defined in
$P_1 \cup \dots \cup P_{m-1}$. By Lemma~\ref{splitting-set-models}, 
since $U_{m}$ is a splitting set, the restriction of the stable model $M_m$ 
to the predicates in $U_{m}$ must coincide with $M_{m-1}$.

Let ${\cal H}_{m,m-1}$ be the set of interpretations of $\mathsf{P}$ such 
that for every interpretation $I \in \mathcal{H}_{m, m-1}$, 
the restriction of $I$ to $U_{m}$ is exactly $M_{m-1}$. 
Since $M_m$ extends $M_{m-1}$, we have $M_m \in \mathcal{H}_{m, m-1}$.
Let ${\cal H}^*_{m,m-1}$ be the set of interpretations of $\mathsf{P}_m^{M_{m-1}}$. 

We define a bijection $\theta: \mathcal{H}_{m, m-1} \to \mathcal{H}^*_{m,m-1}$. For any
interpretation $I \in \mathcal{H}_{m, m-1}$ and any predicate $\mathsf{p}$
defined in $\mathsf{P}_m$ with arguments $\overline{d}$:
$\theta(I)(\mathsf{p}_{\overline{d}}) = \text{true}$ if and only if
$I(\mathsf{p})(\overline{d}) = \text{true}$. This mapping is a bijection because
the atoms of the propositional reduct $\mathsf{P}_m^{M_{m-1}}$ are constructed
specifically for every tuple $\bar{d}$ of possible inputs. Furthermore, $\theta$ is order
preserving \ie if $I_1 \leq I_2$ then $\theta(I_1) \leq \theta(I_2)$ and vice
versa as one can easily verify.

We now show that $I$ is a model of $\mathsf{P}$ extending $M_{m-1}$ if and
only if $\theta(I)$ is a model of the reduct $\mathsf{P}_m^{M_{m-1}}$.
Assume $I \in \mathcal{H}_{m, m-1}$ is a model of $\mathsf{P}$. Suppose, for the
sake of contradiction, that $\theta(I)$ is not a model of the reduct
$P_m^{M_{m-1}}$. 
This implies there exists a rule $r'$ of the form $\mathsf{p}_{\bar{d}} \leftarrow \mathsf{q}^1_{\bar{d_1}}, \dots, \mathsf{q}^k_{\bar{d_k}}$ 
in $\mathsf{P}_m^{M_{m-1}}$, that is not satisfied by $\theta(I)$. That is, 
the body of the rule is true and the head is false.
By the definition of the reduct, this rule was generated
from a rule $r$ in $\mathsf{P}_m$ and a Herbrand state $s$.
The generation of this rule implies that all literals in $r$ which consist of predicates of the first $m-1$ strata
evaluated to true under $M_{m-1}$ and this state $s$. Moreover, the body atoms $q^i_{\bar{d_i}}$ 
in the reduct correspond to literals in $r$ of the form $\mathsf{q}^i(\mathsf{E_1}, \dots, \mathsf{E_{i_a}})$ where $\mathsf{q}^i$ 
is a predicate symbol belonging to the current stratum $\mathsf{P}_m$ while $\mathsf{E_1}, \dots, \mathsf{E_{i_a}}$ are expressions of lower stratum predicates
which evaluate into $\bar{d_i}=d_1,\cdots,d_{i_a}$. 
Since the body is true in $\theta(I)$ for any $\mathsf{q}^i$ it is $\theta(I)(\mathsf{q}^i_{\bar{d_i}}) = \mtrue$ 
and so it must be $I(\mathsf{q}^i)(\bar{d_i}) = \mtrue$ by the definition of $\theta$.
Consequently, the body of $r$ must be true for that specific state $s$. On the other hand, since $\mathsf{p}_{\bar{d}}$ is $\mfalse$
in $\theta(I)$, the head $I(\mathsf{p})\ \bar{d} = \mfalse$ by definition of $\theta$. This contradicts the assumption 
that $I$ is a model of $\mathsf{P}$ since we found a rule and $s$ that is not satisfied in $I$.

Now assume that $\theta(I)$ is a model of $P_m^{M_{m-1}}$ and suppose $I \in \mathcal{H}_{m, m-1}$ is not
a model of $\mathsf{P}$. The argument that leads to the contradiction follows a similar line of thought as the other direction
and is omitted.

From the previous statement it follows that $\theta^{-1}(N_m)$ is a model of
$\Prog$. But since the bijection is order preserving and $N_m$ is minimum model of $\Prog_m^{M_{m-1}}$ then
for any other model $M' \in {\cal H}_{m,m-1}$ $\theta^{-1}(N_m) \leq M'$. Since
$M_m \in {\cal H}_{m,m-1}$ and $M_m$ is minimal then it must be
$M_m=\theta^{-1}(N_m)$ and the statement of the lemma follows.
\end{proof}

Before establishing the complexity bounds for query evaluation, we must first
characterize the size of the underlying semantic domains and the space required
to represent their elements. The following lemma provides an upper bound on the
size and memory needed for storing values of the respective domain.

\begin{lemma}\label{space-bound}
Let $\Prog$ be a Higher-Order Datalog$^\neg$ program and let $n= |U_{\Prog}|$. 
Let $\rho$ be any type of order $i$. Then
$| \mo{\rho}_{U_{\mathsf{P}}} | \leq \exp_{i}(|U_{\mathsf{P}}|^d)$ for some
natural constant number $d$. Furthermore, for any $d \in \mo{\rho}_{U_{\mathsf{P}}}$, we can
store $d$ in memory using at most $\exp_{i-1}(|U_{\mathsf{P}}|^{d'})$ bits (or
tape cells) for some constant $d'$.
\end{lemma}
\begin{proof}
We prove the statement by induction on the order $i$ of the type $\rho$.
If $i=0$ then we have the individual type $\iota$ .
The size of the domain $\mo{\iota}$ is $|U_\mathsf{P}| = n$, which is $\leq \exp_0(n^1)$.
%The size of $\mo{\bool}$ is $2$, which is trivially bounded by $n$, assuming $n \geq 2$.
Any element $d \in \mo{\iota}$ can be uniquely represented by an index requiring $\lceil \log n \rceil$ bits.
An element of $\mo{\bool}$ requires 1 bit.

Assume the hypothesis holds for all types of order $j < i$.
Let $\rho$ be a type of order $i$. Then, $\rho$ is of the form $\rho_1 \to \cdots \to \rho_m \to \bool$,
where each argument type $\rho_j$ has order $k_j < i$ (so $k_j \leq i-1$). The case where $\rho = \bool$ is trivial and omitted.
The domain $\mo{\rho}$ consists of all functions $\mo{\rho_1} \to \cdots \to \mo{\rho_m} \to \mo{\bool}$.
The size of this domain is:
$|\mo{\rho}| = |\mo{\bool}|^{|\mo{\rho_1}| \times \cdots \times |\mo{\rho_m}|}$.
By the induction hypothesis, for each $j$, $|\mo{\rho_j}| \leq \exp_{k_j}(n^{c_j}) \leq \exp_{i-1}(n^{c_j})$ for some constants $c_j$.
The product of these sizes is bounded by $\exp_{i-1}(n^{d'})$ for some constant $d'$.
Consequently,
\[ |\mo{\rho}| \leq 2^{\exp_{i-1}(n^{d'})} = \exp_i(n^{d'}) \]
for defined constant $d'$.
Furthermore, to store an element $f \in \mo{\rho}$, we can use its full function table representation
(since domains are finite). The size of this representation is equal to the number of entries in the table,
which is the size of the input domain:
\(\prod_{j=1}^m |\mo{\rho_j}| \leq \exp_{i-1}(n^{d'}) \)
bits for some constant $d'$.
\end{proof}

\begin{lemma}\label{expression-computation}
Let $\Prog$ be a $(k+1)$-Order stratified linear Datalog$^{\neg}$ program with $k \geq 1$
that defines a generic query and let $\mathsf{E}$ with type $\rho$
be some expression using predicates from $\Prog$ and $\rho$ is of at most $k$ order. 
There exists a deterministic Turing machine
that takes as input an encoding of a database $D$ that uses $n$ individual
constant symbols and the encoding of a state $s$ for the variables and
computes the meaning $\mwrs{\mathsf{E}}{M}{s}$ where $M$ is
the unique two-valued stable model of $\Prog \cup D$. Furthermore, it does so
using at most $exp_{k-1}(n^d)$ tape cells for some constant $d$.
\end{lemma}

\begin{proof}
% For the sake of brevity, in the following we omit a detailed explanation of the
% specific memory encoding for higher-order predicates (or, equivalently,
% higher-order Boolean functions). We will argue instead on the size of these
% representations. Furthermore, the Turing machines (or algorithms) described are
% presented at a higher level of abstraction.

We will prove the statement by induction on the strata of the program 
$\Prog \cup D$. More formally the statement we will show is the following:

Assume that for any expression $\mathsf{E}$ containing only predicates from the first $m$ strata and some Herbrand state
$s$ there exists a Turing machine (an expression evaluator) that computes $\mwrs{\mathsf{E}}{M}{s}$, 
using at most $exp_{k-1}(n^d)$ tape cells for some constant $d$. Then for any expression 
$\mathsf{E}$ containing only predicates from the first $m+1$ strata and some Herbrand state $s$ 
there exists a Turing machine (an expression evaluator) that computes $\mwrs{\mathsf{E}'}{M}{s'}$, 
using at most $exp_{k-1}(n^{d'})$ tape cells for some constant $d'>d$.

We proceed with the base case of the induction. Without loss of generality, we can assume that the first stratum
contains only the external predicates supplied by the input database $D$ in the form of facts.
These are first-order predicates and are not redefined in the query program $\mathsf{P}$.
Thus, we partition by strata the program as $\mathsf{P} \cup D = D \cup \mathsf{P}_1 \cup \cdots \cup \mathsf{P}_{m_{max}}$.
With this assumption the base of the induction is trivial to prove. We proceed with the induction step.

Let $M$ denote the unique stable model of $\mathsf{P} \cup D$. For each $m \leq
m_{max}$, we define $M_m$ to be the unique stable model of $D \cup \mathsf{P}_1 \cdots \cup
\mathsf{P}_m$ and  $U_{m+1}$ the splitting set that corresponds to the predicates of the first $m$
strata such that  $b_{U_{m+1}}(\mathsf{P}) = D \cup \mathsf{P}_1 \cup \cdots \cup
\mathsf{P}_m$. Notice that each constant of type $\iota$ appearing in $\Prog$
must appear in the external predicates of database $D$ since the program defines
a generic query (\ie $\mathsf{P}$ contains no individual constants by itself). Therefore,
for each $m \leq m_{max}$, $b_{U_{m+1}}(\mathsf{P})$ contains every individual constant
and so each $b_{U_{m+1}}(\mathsf{P})$ shares the same Herbrand universe.
By Lemma~\ref{splitting-set-models} the restriction of $M_{m+1}$
with respect to $U_{m+1}$ is the unique stable model of $b_{U_{m+1}}(\mathsf{P})$
and thus coincides with $M_m$. 

We define  helper routines which we call the ``query machines'' 
which we will then use to describe the expression evaluator.

\paragraph{\bf The query machines:}
For any $1 \leq m' \leq m+1$ we define a Turing machine that computes each query of the form 
$M_{{m'}}(\mathsf{p})\ d_1,\ldots,d_j$ for $\mathsf{p} \in \mathsf{P}_{m'}$ 
when given the input database encoded in some sensible form and a set of 
$d_i \in \mo{\rho_i}$ as a function (or a higher-order set) representation and does
so using at most $\exp_{k-1}(poly(n))$ tape cells. Notice that since the types of $\rho_i$ are
used as arguments and thus are of order at most $k$, writing down those elements
in the input takes space at most $j \times exp_{k-1}(poly(n))$ by Lemma~\ref{space-bound}.

By Lemma~\ref{reduct-models} it is enough to compute the least model of the
propositional program $\mathsf{P}^{M_{m'-1}}_{m'}$ and use the bijection to determine the truth value of 
$M_{{m'}}(\mathsf{p})\ d_1,\ldots,d_j$ for $\mathsf{p} \in \mathsf{P}_{m'}$. 
But the machine cannot instantiate this program since each argument has a domain of cardinality
of up to  $\exp_{k}(poly(n))$ thus the set
of all different propositional atoms and rules will exceed the memory limit. 
The machine instead operates in a top-down fashion instantiating the rules as needed. The equivalence of this procedure 
with the bottom-up computation is well-known in literature for propositional programs.
The machine operates as follows:

\begin{enumerate}
\item It starts with the atom $(\mathsf{p}\,d_1\cdots d_n)$ written on the tape, where the $d_i$ are
      representations of expression values, \ie a string that describes the value\footnote{For example, a relation can be represented as a string listing all tuples in some predetermined order.}.
\item \label{select-step-appendix}Non-deterministically selects a rule
      $\mathsf{p}\ \mathsf{R}_1 \cdots \mathsf{R}_j \lrule \mathsf{L}_1, \dots, \mathsf{L}_r$
      from $\Prog_{m'}$ and guesses a valid assignment for the rule's variables such that
      $d_i$ is assigned to $\mathsf{R}_i$ for all $i \in \{ 1,\dots, j\}$.
\item Iterates through the literals $\mathsf{L}_i$ in the rule body:
\begin{itemize}
  \item If $\mathsf{L}_i = ({\tt not}\ \mathsf{E})$ is a negative literal then the machine computes
        $\mwrs{\mathsf{L}_i}{M_{m'-1}}{s[\bar{\tt R}/\bar{d}]}=\mwrs{\mathsf{L}_i}{M}{s[\bar{\tt R}/\bar{d}]}$
        by invoking an expression evaluator for the $m'-1$ strata defined for $\mathsf{E}$
        to compute the value of $\mwrs{\mathsf{E}}{M}{s[\bar{\tt R}/\bar{d}]}$. The existance of this evaluator
        is given by the induction hypothesis since $m'-1 \leq m$. Then it performs the negation operator to the output. 
        The machine that is invoked uses at most $\exp_{k-1}(poly(n))$ space which is freed after the computation.
  \item If $\mathsf{L}_i =\mathsf{E}$ is a positive literal of no predicates from stratum $m'$ then the machine computes
        $\mwrs{\mathsf{L}_i}{M_{m'-1}}{s[\bar{\tt R}/\bar{d}]}=\mwrs{\mathsf{L}_i}{M}{s[\bar{\tt R}/\bar{d}]}$
        by invoking an expression evaluator for $m'-1$  strata defined for $\mathsf{E}$  
        to compute the value of $\mwrs{\mathsf{E}}{M}{s[\bar{\tt R}/\bar{d}]}$. The existence of this evaluator
        is given by the induction hypothesis since $m'-1 \leq m$.
        The machine that is invoked uses at most $\exp_{k-1}(poly(n))$ space which is freed after each computation.
  \item If there is a predicate in $\mathsf{L}_i$ from the $(m')$-th stratum, then this is the unique
        recursive literal $\mathsf{q}\ \mathsf{E}_1 \cdots \mathsf{E}_t$ implied by linearity.
        Every $\mathsf{E}_i$ is of order at most $k$ and includes predicates from
        the first $m$ strata. Iterate through $\mathsf{E}_1,\dots,\mathsf{E}_t$ invoking each time the corresponding evaluator,
        passing as parameter the guessed variable assignment, to compute the value that corresponds to each $  \mathsf{E}_i$ 
        and store the representations $d'_1,\dots,d'_t$.
        Construct a new ground atom $\mathsf{q}\,d'_1 \cdots d'_t$ with those representations. The memory used for this atom 
        is bounded by $\exp_{k-1}(poly(n))$ by Lemma~\ref{space-bound} given the order of the argument types.
\end{itemize}
\item If all literals $\mathsf{L}_i$ evaluate to $\mtrue$ without encountering a recursive literal,
      then evaluate the query result as true.
\item Otherwise, if one of the literals $\mathsf{L}_i$ is a recursive one, overwrite the current goal with $\mathsf{q}\,d'_1 \cdots d'_t$
      and loop back to Step~\ref{select-step-appendix}. The size of goal never exceeds $\exp_{k-1}(poly(n))$ cells.
\end{enumerate}
Note that the aforementioned query machine is non-deterministic and may contain
computational paths that do not terminate (\eg consider the rule
\lstinline`p:-p`). So, this machine is not technically a decider of whether
$M_{{m'}}(\mathsf{p})\ d_1,\ldots,d_j$ for $\mathsf{p} \in \mathsf{P}_{m'}$ is true or false. 
However, for machines that are
space bounded (more specifically by a space constructible function
$f(n)$~\citep{Sipser}) there exists a corresponding non-deterministic Turing
machine where every computational path terminates. In particular, the machine includes an initial step that given the input size $n$ 
it computes the the value $\exp_{k}(poly(n))$ in $\exp_{k-1}(poly(n))$ space and uses this isolated space to simulate
a step counter. Each computation path will then auto-reject if the number of steps has exceeded the number 
of possible configurations.

Finally, by~\cite{Savitch1970}, there
also exists a deterministic machine that decides whether $M_{{m'}}(\mathsf{p})\ d_1,\ldots,d_j$ for $\mathsf{p} \in \mathsf{P}_{m'}$
is true or false in the same space complexity class. In the following text 
when we refer to the ``query'' machine we mean the deterministic version.

\paragraph{\bf The expression evaluator:}
We now show how to extend this computation to any arbitrary fixed expression
$\mathsf{E}$ of order at most $k$ that contains predicates of the first $m+1$
strata while still using memory that is bounded by $\exp_{k-1}(poly(n))$. To be entirely formal the calculation of 
memory usage should also account for the length of the expression. For reasons of brevity, we avoid to account for that
since we are interested in how bounded expressions (\eg those appearing in the query program) 
behave under variable input database sizes.

The evaluator machine for an expression $\mathsf{E}$ will invoke in sequence the previous query machines when needed.
We show the existence of such machine and argue about the memory used by induction on the length 
of the expression. For the base
case, assume that the length of $\mathsf{E}$ is 1. It is easy to verify that evaluators using at most $\exp_{k-1}(poly(n))$ space exist 
for all expressions of length 1. We have two cases:
\begin{enumerate}
\item $\mathsf{E}=\mathsf{X}$ for some variable or $\mathsf{E}=\mathsf{c}$ for some $\iota$ constant.
\item $\mathsf{E}=\mathsf{p}$ for some constant predicate $\mathsf{p}$.
\end{enumerate}
The first case is trivial since we use the variable assignment.
For the second case, let $\mathsf{p} \in \mathsf{P}_{m'}$ for some $m' \leq m+1$.
By the order assumption of $\mathsf{E}$, $\mathsf{p}$ has a type of order at most $k$. 
We generate the different tuples of possible arguments 
in an increasing manner. For each combination we invoke the 
$m'$-query machine described previously to get the corresponding truth value forming a set of size $\exp_{k-1}(poly(n))$ in memory.
The query machine itself operates in $\exp_{k-1}(poly(n))$ space that is
reused for each separate query and  the total result is stored in $\exp_{k-1}(poly(n))$ memory.

We proceed with the general step of the induction. Assume that for every 
expression of length at most $l$ there exist evaluators that compute it
in at most $\exp_{k-1}(poly(n))$ tape cells. Let $\mathsf{E}$ be an expression
of size $l+1$. We have two cases:
\begin{enumerate}
\item $\mathsf{E}=\mathsf{X}\ \mathsf{E}_1\ \mathsf{E}_2 \ldots \mathsf{E}_t$ for some variable $\mathsf{X}$.
\item $\mathsf{E}=\mathsf{q}\ \mathsf{E}_1\ \mathsf{E}_2 \ldots \mathsf{E}_t$ for some predicate $\mathsf{q}$.
\end{enumerate}
Both cases are easy to argue about. We show the second. Let $\mathsf{q} \in \mathsf{P}_{m'}$ with $m' \leq m+1$.
Notice how for each $\mathsf{E}_i$ it is an argument and thus the order of its type is at most $k$. We perform the following steps:
\begin{enumerate}
\item Evaluate the expressions $\mathsf{E}_1, \dots, \mathsf{E}_t$ sequentially
      to obtain and store their representations $d_1,\dots,d_t$. Since every $\mathsf{E_i}$ 
      has length at most $l$ and type of order at most $k$,
      by inductive hypothesis each can be computed by an appropriate evaluator 
      and stored using total memory bounded by $\exp_{k-1}(poly(n))$.
\item To determine the representation of the result, the machine enumerates all possible
      tuples  $d_{t+1}, \dots, d_{t'}$ in increasing fashion regarding the remaining arguments.
      For each tuple, it invokes the query machine for the $m'$-th stratum with the goal
      $(\mathsf{q}\, d_1 \cdots d_{t}\, d_{t+1} \cdots d_{t'})$ in reusable memory.
\end{enumerate}
The bound on the representation formed is $\exp_{k-1}(poly(n))$ that follows from
the order of $\mathsf{E}$ and Lemma~\ref{space-bound}. 
%
%
% Each $\mathsf{E}_i$
% appears in an argument position therefore has an order of at most $k$ and has
% length less or equal to $l$. By inductive hypothesis each can be computed and
% stored in at most $\exp_{k-1}(poly(n))$ tape cells. Let $\mathsf{q}$ have a type
% of $\rho= \rho_1 \to \cdots \to \rho_{t'} \to \bool$ with $t' \geq t$. For any
% $i,t<i \leq t'$ since the expression $\mathsf{E}$ has a type of order at most
% $k$ then each $\rho_i$ has order of at most $k-1$. For the fixed set of
% initially computed arguments $d_1,d_2,\ldots,d_t$ completed with every possible
% set of possible arguments $d_{t+1},\ldots,d_{t'}$ we invoke the query machine
% described previously. The bound on the total memory size follows immediately.
\end{proof}

\begin{retheorem}{upper-bound-theorem}
Let $\Prog$ be a $(k+1)$-order Stratified Linear  Datalog$^{\neg}$ program that defines a
query. Then, there exists a deterministic Turing machine that takes as input an
encoding of a database $D$ that uses $n$ individual constant symbols and a
ground atom $p(\bar{a})$, where $p$ is a predicate constant of $\Prog$ and
$\bar{a}$ is a tuple of those individual constants, and decides whether
$p(\bar{a}) \in \mathcal{Q}_\mathsf{P}(D)$, while using at most $\exp_{k-1}(n^d)$
tape cells for some constant $d$.
\end{retheorem}
\begin{proof}
Immediate from Lemma~\ref{expression-computation}.
\end{proof}
\section{The Graph Recoloring Problem and its Transformation to QBF}\label{recoloring-qbf-eval}

The theoretical alignment between the second-order fragment of Stratified Linear
Higher-Order Datalog and \textsf{PSPACE} naturally motivates compiling these
programs into Quantified Boolean Formulas (QBFs), allowing highly optimized
modern solvers to serve as the execution backend. Rather than presenting a
systematic treatment, this section serves as a proof of concept by walking
through the Graph Recoloring program presented in
Example~\ref{graph_recoloring_example}.

To intuitively illustrate the mechanics of this translation, we break the
process down into three phases: grounding the first-order variables
into boolean vectors, compiling the lower strata into a boolean
formula, and utilizing universal quantification to compress the linear path
search.

\subsection{Grounding the Configuration Space}
In first-order Datalog, the active state of a computation is typically a single
node or a tuple of domain elements. In second-order Datalog, the active state is
an entire relation.

Consider an input graph with $N$ vertices (nodes) and a predefined set of $K$
available colors. The configuration of the board at any given moment is fully
captured by the existentially quantified first-order binary relation $C$, where
$C(X, Col)$ evaluates to true if vertex $X$ is painted with color $Col$. Because
the domain of vertices and colors is strictly finite, the total number of
possible configurations is also finite.

We can systematically flatten this first-order relation into a one-dimensional
boolean vector $\vec{c}$ of length $N \times K$. Each bit $c_{x, col}$ in this
vector represents the truth value of the corresponding atom $C(X, Col)$.
Therefore, the dynamic evaluation of the 2nd-order Datalog program is
mathematically equivalent to traversing a massive, albeit finite, graph of size
$2^{N \times K}$, where each node is a specific boolean bit-vector representing
a distinct coloring state.

\subsection{Compiling the Static Strata into Formulas}

The strict stratification of the language (Definition~\ref{stratified}) guarantees that rules
are evaluated in a feed-forward manner without negative cycles. Because they
operate entirely on known, static structures (the graph edges $E$, the vertices
$V$, the colors $K$, and the boolean vectors), they act purely as declarative
constraints. A compiler can systematically map each of these rules directly into
a static, quantifier-free boolean formula.

First, we define the boolean formulas ensuring a vector $\vec{c}$ represents a
valid graph coloring based on the logic:
\begin{align*}
T_{\mathit{conflict}}(\vec{c}) &= \bigvee\nolimits_{(x,y) \in E} \bigvee\nolimits_{col \in K} (c_{x, col} \land c_{y, col}) \\
T_{\mathit{multiple\_colors}}(\vec{c}) &= \bigvee_{x \in V} \bigvee\nolimits_{col_1, col_2 \in K} (c_{x, col_1} \land c_{x, col_2} \land col_1 \neq col_2) \\
T_{\mathit{uncolored}}(\vec{c}) &= \bigvee\nolimits_{x \in V} \bigwedge\nolimits_{col \in K} \neg c_{x, col} \\
T_{\mathit{invalid\_coloring}}(\vec{c}) &= T_{conflict}(\vec{c}) \lor T_{multiple\_colors}(\vec{c}) \lor T_{uncolored}(\vec{c}) \\
T_{\mathit{valid\_coloring}}(\vec{c}) &= \neg T_{\mathit{invalid\_coloring}}(\vec{c})
\end{align*}

Next, we define the formulas ensuring the transition from state $\vec{c}_1$ to
state $\vec{c}_2$ modifies exactly one vertex. We first define a helper formula
$T_{\mathit{diff}}(x, \vec{c}_1, \vec{c}_2)$ which is true if vertex $x$ loses a color it
had in $\vec{c}_1$ (implying a change between the states):
\begin{align*}
T_{\mathit{diff}}(x, \vec{c}_1, \vec{c}_2) &= \bigvee\nolimits_{col \in K} (c_{1, x, col} \land \neg c_{2, x, col}) \\
T_{\mathit{has\_diff}}(\vec{c}_1, \vec{c}_2) &= \bigvee\nolimits_{x \in V} T_{\mathit{diff}}(x, \vec{c}_1, \vec{c}_2) \\
T_{\mathit{multiple\_diffs}}(\vec{c}_1, \vec{c}_2) &= \bigvee\nolimits_{x, y \in V} (T_{\mathit{diff}}(x, \vec{c}_1, \vec{c}_2) \land T_{\mathit{diff}}(y, \vec{c}_1, \vec{c}_2) \land x \neq y) \\
T_{\mathit{one\_diff}}(\vec{c}_1, \vec{c}_2) &= T_{\mathit{has\_diff}}(\vec{c}_1, \vec{c}_2) \land \neg T_{\mathit{multiple\_diffs}}(\vec{c}_1, \vec{c}_2)
\end{align*}

Finally, by conjoining these constraints, the compiler generates
the formula $T_{\mathit{valid\_step}}(\vec{c}_1, \vec{c}_2)$ for the \lstinline`valid_step` predicate:
\begin{equation*}
T_{\mathit{valid\_step}}(\vec{c}_1, \vec{c}_2) = T_{\mathit{valid\_coloring}}(\vec{c}_1) \land T_{\mathit{valid\_coloring}}(\vec{c}_2) \land T_{\mathit{one\_diff}}(\vec{c}_1, \vec{c}_2)
\end{equation*}

% This formula acts as the ``hardware verification circuit'' of our state machine.
% It evaluates to \textit{true} if and only if the vector $\vec{c}_2$ represents a
% perfectly valid single-token recoloring step from the vector $\vec{c}_1$.

\subsection{Compressing the Linear Path via Savitch's Theorem}
The final stratum of our program introduces the predicate
\lstinline`recolorable` which is recursive. Because the rule contains at most
one recursive call to itself, the evaluation does not branch into an
exponentially expanding search tree.

Given our configuration space, the maximum possible length of a non-looping path
is bounded by $2^{N \times K}$. Naively asking a boolean solver to search for a
path of this length would require unrolling the formula $T_{\mathit{valid\_step}}$
exponentially many times. For instance, testing a path of length $m$ requires
writing $\exists \vec{v}_1 \dots \exists \vec{v}_m [T_{\mathit{valid\_step}}(\vec{start}, \vec{v}_1)
\land \dots \land T_{\mathit{valid\_step}}(\vec{v}_{m-1}, \vec{target})]$. If $m = 2^{N \times K}$, the
resulting formula would be exponentially larger than the memory capacity of any
physical machine, rendering the translation useless.

To prevent this exponential blow-up, the compiler leverages the core principle
of Savitch's theorem, utilizing the universal quantifiers of QBF to
act as a spatial compression mechanism. We recursively define a bounded path
formula $P^i_{\mathit{recolorable}}(\vec{u}, \vec{v})$ which evaluates to true if there is a valid
path from $\vec{u}$ to $\vec{v}$ of length at most $2^i$.

For the base case:
$P^0_{\mathit{recolorable}}(\vec{u}, \vec{v}) \equiv (\vec{u} \approx \vec{v}) \lor T_{\mathit{valid\_step}}(\vec{u}, \vec{v})$.
For the recursive step:
$$P^i_{\mathit{recolorable}}(\vec{u}, \vec{v}) = \exists \vec{m} \forall \vec{a} \forall \vec{b} \left[ \left( (\vec{a}\approx\vec{u} \land \vec{b}\approx\vec{m}) \lor (\vec{a}\approx\vec{m} \land \vec{b}\approx\vec{v}) \right) \rightarrow P^{i-1}_{\mathit{recolorable}}(\vec{a}, \vec{b}) \right]$$

In this recursive definition, we posit the existence of an intermediate halfway
state $\vec{m}$. Instead of explicitly writing the formula $P^{i-1}_{\mathit{recolorable}}$ twice
(once for the left half $\vec{u} \to \vec{m}$ and once for the right half
$\vec{m} \to \vec{v}$), the universal quantifiers $\forall \vec{a}, \vec{b}$ act
as a structural umbrella. They force the solver to verify both halves of the
journey using only a single, mathematically compressed copy of the $P^{i-1}_{\mathit{recolorable}}$
sub-formula.

% By expanding this definition up to the required logarithmic depth ($i = N \times K$),
% the compiler generates a final QBF query of polynomial size $O((N \times K)^2)$.

% In conclusion, the linearity of the Datalog fragment ensures the problem is a
% pure path search, its stratification allows the static compilation of the
% transition matrix, and its second-order nature bounds the space to dimensions
% that Savitch's algorithm can neatly compress. This pipeline successfully shifts
% the computational burden from a hypothetical higher-order declarative engine to
% highly optimized, space-efficient QBF solvers.

% \begin{thebibliography}{1}
% \bibitem{Bonsma2009}
% P. Bonsma and L. Cereceda,
% ``Finding paths between graph colourings: PSPACE-completeness and superpolynomial distances,''
% \textit{Theoretical Computer Science}, vol. 410, pp. 5215-5226, 2009.
% \end{thebibliography}
\fi
\end{document}